\documentclass[]{interact}  

\usepackage{apacite}
\usepackage{graphics} 
\usepackage{epsfig} 
\usepackage{mathptmx} 
\usepackage{times} 
\usepackage{xcolor}
\usepackage{amsmath} 
\usepackage{amssymb}  
\usepackage{subcaption}
\usepackage[english]{babel}
\newtheorem{theorem}{Theorem}[section]

\newtheorem{lemma}[theorem]{Lemma}
\theoremstyle{definition}
\newtheorem{definition}{Definition}[section]
\newtheorem{assumption}{Assumption}
\newtheorem{remark}{Remark}

\begin{document}


\title{Simultaneous Identification and Optimal Tracking Control of Unknown Continuous Time Nonlinear System With Actuator Constraints Using Critic-Only Integral Reinforcement Learning}

\author{
\name{Amardeep Mishra\textsuperscript{a}\thanks{Amardeep Mishra Email: ae15d405@smail.iitm.ac.in} and Satadal Ghosh\textsuperscript{b}\thanks{Satadal Ghosh Email: satadal@iitm.ac.in}}
\affil{\textsuperscript{a}Ph.D. Student, Department of Aerospace Engineering, IIT Madras, Chennai, India-600036; \textsuperscript{b}Faculty, Department of Aerospace Engineering, IIT Madras, Chennai, India-600036}
}

\maketitle
\begin{abstract}
In order to obviate the requirement of drift dynamics in adaptive dynamic programming (ADP), integral reinforcement learning (IRL) has been proposed as an alternate formulation of Bellman equation.
However control coupling dynamics is still needed to obtain closed form expression of optimal control effort.
In addition to this, initial stabilizing controller and two sets of neural networks (NN) (known as Actor-Critic) are required to implement IRL scheme.
In this paper, a stabilizing term in the critic update law is leveraged to avoid the requirement of an initial stabilizing controller in IRL framework to solve optimal tracking problem with actuator constraints.
With such a term, only one NN is needed to generate optimal control policies in IRL framework.
This critic network is coupled with an experience replay (ER) enhanced identifier to obviate the necessity of control coupling dynamics in IRL algorithm.
The weights of both identifier and critic NNs are simultaneously updated and it is shown that the ER-enhanced identifier is able to handle parametric variations better than without ER enhancement.
The most salient feature of the novel update law is its variable learning rate, which scales the pace of learning based on instantaneous Hamilton-Jacobi-Bellman (HJB) error. 
Variable learning rate in critic NN coupled with ER technique in identifier NN help in achieving tighter residual set for state error and error in NN weights as shown in uniform ultimate boundedness (UUB) stability proof.
The simulation results validate the presented "identifier-critic" NN on a nonlinear system.
\end{abstract}

\section{INTRODUCTION}
Optimal control aims to find control policies that minimizes an objective function subjected to plant dynamics.
It can be obtained using either Pontryagin's minimum principle or solving HJB equation. Traditionally, these schemes are off-line and require complete knowledge of system dynamics to find the control structure before implementation.
Solution of HJB equation, which is a nonlinear partial differential equation (PDE), provides the optimum value function that can be utilized to generate the optimum control policies. However, for generic nonlinear system solving HJB equation is often intractable. 
In order to by-pass the challenge of solving HJB equation directly, considerable effort has been dedicated in literature to finding algorithms that provide approximate solutions to HJB equations, such as, iterative Approximate Dynamic Programming methods. 
First few results utilizing adaptive dynamic programming (ADP) in optimal regulation problem for generic continuous time nonlinear systems (CTNS) were presented in 
\cite{abu2005nearly}, 
\cite{vamvoudakis2010online}.  
A recursive least square formulation was presented to tune the NN weights in \cite{abu2005nearly}, whereas,  
\cite{vamvoudakis2010online} proposed synchronous tuning of actor-critic neural networks (NN) in their paper. 
Similarly, \cite{dierks2010optimal} and \cite{zhang2011data} were first few papers on optimal tracking problem utilizing ADP algorithm.
However, in both these papers, they needed two different controllers, i.e., a transient and a steady state controller to implement optimal tracking.
Further, the steady state controller required inverting control gain matrix.

In order to address this issue, \cite{modares2014integral}, \cite{kiumarsi2014reinforcement} and \cite{modares2014optimal} proposed the concept of augmented system states comprising of error and desired system states.
In their formulation, the optimal tracking controller was the one that minimised the cost function subjected to the nominal augmented dynamics.
By using augmented states, they could generate optimal tracking controller without requiring the invertibility of control gain matrix.

Most of the schemes discussed above except \cite{dierks2010optimal} require an initial stabilizing control to initiate the process of policy iteration and two different NNs to implement reinforcement learning (RL). Finding an initial stabilizing controller is often not a trivial task. 
The criteria of initial stabilizing control for ADP algorithm was relaxed in \cite{dierks2010optimal} wherein they proposed a modified update law that contained a stabilizing term, which used to come into effect when Lyapunov function was non decreasing along the system trajectories.
Similarly, \cite{yang2015robust} leveraged the stabilizing term proposed in \cite{dierks2010optimal} to develop an ADP algorithm for robust optimal tracking control of nonlinear systems that did not require an initial stabilizing controller.
Additionally, they had robust terms to counter any variations in drift dynamics.
In both \cite{dierks2010optimal} and \cite{yang2015robust}, tracking control action could be generated by only a single NN i.e., critic NN. 
However, their method required the knowledge of nominal plant dynamics and did not include any actuator constraint. 

When information of nominal plant dynamics is not known, then identifiers have been utilized in literature to approximate the dynamics before using it in ADP algorithm.
They can be broadly classified into two categories, i.e., (a) system identification when the structure of the dynamics is unknown, and (b) system identification when structure of the dynamics is known, but the parameters of the plant are unknown.
First few papers utilizing identifiers with ADP algorithms are, \cite{bhasin2013novel}, \cite{yang2014reinforcement} and \cite{lv2016online} for regulation and  \cite{zhang2011data}, \cite{na2014approximate} and \cite{hou2017adaptive} for tracking.
In \cite{zhang2011data}, the knowledge of control coupling dynamics was assumed to be unknown, and the identifier was run prior to ADP algorithm. On the other hand, in \cite{bhasin2013novel} the identifier was run simultaneously with the ADP algorithm under the assumption of prior knowledge of control coupling dynamics. Using NN-based identifiers, \cite{yang2014reinforcement} implemented optimal regulation, where even the structure of the dynamical system was unknown. 
An ER technique-based identifier was presented in \cite{modares2013adaptive} to minimize the difference between actual state and identifier state.
In all the aforementioned papers, identifier state $\hat{x}$ converges to actual state $x$, however, the convergence of estimated weights $\hat{W}_I$ to true weights $W_I$ is not guaranteed. 
In order to remedy that, a novel online identification method that ensured convergence to true NN weights based on the weights error, instead of state error was proposed in \cite{lv2016online}. 

Their update law could directly minimize the error between ideal and estimated NN weights, whereas in other schemes presented above, the parameters are tuned such that the difference between identifier output and actual plant output is minimized.
In most of the above schemes, identifiers were utilized to approximate both drift and/or control coupling dynamics, and then the estimated drift and/or control coupling dynamics were then used in the adaptation law for critic and/or actor NN(s) for generating optimal control policy. 
It is important to note that because of the inherent approximation errors in identified drift and control coupling dynamics, the approximation error enters the ADP algorithm via two channels i.e., both drift dynamics and control coupling dynamics, when both of them are unknown \textit{a-priori}.



In ADP schemes mentioned above, the information of nominal plant dynamics was needed either explicitly or from an identifier.
In order to reduce the dependence of ADP algorithms on plant dynamics, integral reinforcement learning (IRL) schemes have been proposed in literature in recent times as an alternate form of Bellman equation to obviate the requirement of drift dynamics. 
Hence, in IRL scheme, drift dynamics does not appear in either adaptation law or the policy improvement stage.
In certain formulations like off-policy IRL methods, such as those presented in \cite{modares2015h}, \cite{zhu2016using}, \cite{zhang2017finite},  \cite{liu2019analysis} and \cite{mishra2020texthinfty}, even the control coupling dynamics is not required to compute optimal policies and are hence model free. 
However, the exploration phase inherent to these off-policy algorithms make it unsuitable for various engineering applications with fast dynamics such as control of aerial vehicles.
The final control policies learnt using aforementioned off-policy IRL methods cannot effectively cope up with either an unpredictable change in trajectory arising out of several reasons like obstacle avoidance requirements or a sudden change in dynamics, such as change in mass or inertia. 
These are some of the prime motivations to consider online and on-policy IRL algorithm in this paper.
To that end, \cite{modares2014integral} and \cite{vamvoudakis2014online} are some of the earliest papers on synchronous tuning of actor-critic NN, based on a gradient descent-driven update law in IRL framework for optimal tracking control problem (OTCP) and regulation problem, respectively. 
\cite{mishra2019criticonly} presented an IRL tracking controller that did not require an initial stabilizing controller and could be implemented with only critic NN. 
However, note that in most of the On-policy IRL formulations for OTCP of CTNS present in the literature, prior knowledge of control coupling dynamics is still needed. 
To the best of authors' knowledge, there is no study on identifier-augmented IRL formulations.

To this end, inspired by IRL formulation in \cite{modares2014optimal}, use of stabilizing term in critic NN update law in \cite{dierks2010optimal} and NN-based parameter identifier in \cite{lv2016online}, this paper expands over \cite{mishra2019criticonly} and presents an identifier-aided IRL algorithm for OTCP of CTNS with unknown model parameters requiring no initial stabilizing controller. Salient features of the control algorithm presented in this paper are as follow.

\begin{enumerate}
    \item In the presented formulation, the control coupling dynamics identifier adds to that in \cite{lv2016online} by leveraging experience replay (ER) technique, thus improving the convergence properties of the identifier NN weights by utilizing the past observations effectively, while a variable gain gradient descent-based critic update law containing stabilizing term is used in the IRL framework similar to that in \cite{mishra2019criticonly}. 
    \item Stability analysis of the synchronous tuning of the ER-augmented identifier and the critic NNs show that the presented scheme yields tighter residual sets for state error and error in NN weights.
    \item Unlike identifier-aided RL-based ADP formulations like \cite{lv2016online} that uses both estimated drift and control coupling dynamics in ADP algorithm, this paper uses only the estimated control coupling dynamics, thus resulting in lesser approximation errors coming from identifier. Thus, the formulation presented in this paper leverages the advantages of both the identifiers and the IRL formulations in a justified way.
    \item Unlike traditional IRL formulations like \cite{modares2014optimal} and \cite{vamvoudakis2014online}, this paper leverages stabilizing term in critic update law in the IRL formulation similar to the one used for RL-based ADP algorithm in \cite{liu2015reinforcement}. This ensures that no initial stabilizing controller is required even in the absence of any prior knowledge of drift dynamics.
    \item Moreover, the use of variable gain gradient descent (similar to \cite{mishra2019variable}) in the critic update law in the presented IRL formulation helps in achieving tighter residual set compared to constant learning rate gradient descent while attaining high learning rate.
\end{enumerate}
The rest of the paper is organized as follows. 
Section \ref{identif} introduces the identifier NN. 
It is divided into two subsections, namely, preliminaries and the ER-enhanced parameter update law and its stability analysis.
Section \ref{prelim} discusses the problem of OTCP for CTNS with actuator constraints and approximation of value function using a single NN, Section \ref{vargain} presents the notion of IRL algorithm for OTCP.
Section \ref{paramupd} provides the parameter update law for critic NN in IRL framework and detailed discussion of the stability when identifier and critic NN are tuned simultaneously. 
Section \ref{res} provides the numerical simulation results to show the effectiveness of the proposed scheme on nonlinear system and finally concluding remarks are presented in Section \ref{conclusion}.

\section{Enhanced system identification using ER}\label{identif}
\subsection{Preliminaries of system identification}
In order to address the issue of non-availability of control coupling dynamics for online and on-policy IRL algorithm, an improved version of online adaptive identifier \cite{lv2016online} is developed in this section utilizing the notion of experience replay (ER). Dynamics of a control-affine system is given as,
\begin{equation}
\dot{x}=f(x)+g(x)u
\label{eq:affine1}
\end{equation}
where, $x \in \mathbb{R}^n$, $u \in \mathbb{R}^m$, $f(x):\mathbb{R}^n \to \mathbb{R}^n$ and $g(x):\mathbb{R}^n \to \mathbb{R}^{n\times m}$. 
Drift and control dynamics are assumed to be Lipschitz continuous in $x$ over compact set $\Omega \subset \mathbb{R}^n$ and hence, can be approximated by NNs  \cite{ren2009neural}. 
Assuming that there exist ideal NN weights that can accurately approximate both $f(x)$ and $g(x)$ as:
\begin{equation}
\begin{split}
f(x)=w_1\xi_1(x)+\epsilon_f;~~g(x)=w_2\xi_2(x) +\epsilon_g
\end{split}
\label{approx}
\end{equation}
where, $w_1 \in \mathbb{R}^{n\times k_{w_1}}$ and $w_2 \in \mathbb{R}^{n\times k_{w_2}}$ are the unknown optimal weights that can accurately approximate the unknown dynamics. 
Regressors for drift and control coupling dynamics are denoted as $\xi_1 \in \mathbb{R}^{k_{w_1}}$ and $\xi_2 \in \mathbb{R}^{k_{w_2}\times m}$, respectively, and $\epsilon_f \in \mathbb{R}^n$ and $\epsilon_g \in \mathbb{R}^{n\times m}$ are the corresponding approximation errors, respectively. According to Weirstrauss higher-order approximation theory \cite{abu2005nearly} and \cite{finlayson2013method}, as the size of regressors i.e., $\xi_1$ and $\xi_2$ increase, i.e $k_{w_1} \to \infty, k_{w_2} \to \infty$, the approximation error goes to zero. Now, using (\ref{approx}) in (\ref{eq:affine1}),
\begin{equation}
 \dot{x}=W_1^T\Phi(x,u)+\epsilon_{T} 
 \label{identifier}
\end{equation}
where, $W_1=[w^T_1;w^T_2] \in \mathbb{R}^{(k_{w_1}+k_{w_2})\times n}$ 
is the combined weight matrix for drift and control dynamics and $\Phi(x,u)=[\xi_1^T(x),u^T\xi_2^T(x)]^T \in \mathbb{R}^{k_{w_1}+k_{w_2}}$ is the combined regressor. 
And, $\epsilon_{T}=\epsilon_f+\epsilon_gu$ represents the combined approximation error.

Now, since ideal identifier weights are unknown, their updated values will be used instead to generate the drift and identified dynamics, i.e.,
\begin{equation}
\hat{f}(x)=\hat{w}_1\xi_1;~\hat{g}(x)=\hat{w}_2\xi_2
\label{eq:fgha}
\end{equation}
Therefore, the identified dynamics can be represented as, 
\begin{equation}
\dot{\hat{x}}=\hat{W}^T_1\Phi(x,u)
\label{eq:final_xhat}
\end{equation}
where, $\hat{x}$ is the state of the identifier.
\subsection{ER-based parameter update law for identifier}

Various online parameter update schemes for identifiers have been presented in literature like minimization of the residual identifier output error 
in \cite{zhang2011data}, modified robust
integral of sign of the error (RISE) algorithm in \cite{bhasin2013novel}, experience replay (ER)-based method in \cite{modares2013adaptive}, etc. However, in all the aforementioned papers, it was desired to achieve the identifier state $\hat{x}$ to converge to actual state $x$, while convergence of estimated weights $\hat{W}_1$ to true weights $W_1$ was not guaranteed. As a remedy to this problem, a promising online identification method that ensured convergence to true NN weights based on the weights error, instead of state error, was proposed in \cite{lv2016online}. An ER-based augmented version of the identifier update law in \cite{lv2016online} is now presented in this section.
Note that the update law presented here has additional advantages compared to the one developed in \cite{lv2016online}, for instance, ER leads to efficient utilization of the past observations in learning NN weights. It should be noted here that the ER-based identification scheme presented in \cite{modares2013adaptive} is different from the one in this paper in the sense that in \cite{modares2013adaptive} error between actual state of the system and identifier state was minimized, whereas, in this paper the error between estimated and ideal identifier NN weights is minimized.

In order to develop the update law, low-pass-filtered versions of regressor vector $(\Phi_f)$  and state vector $(x_f)$ are defined as follows, where the $f$ subscript denotes the filtered variables. 

\begin{equation}
\begin{split}
k\dot{\Phi}_f+\Phi_f=\Phi;~~k\dot{x}_f+x_f=x
\end{split}
\label{filt_reg}
\end{equation}
where, $k>0$ is a scalar.
Next, two matrices $\Pi$ and $K$ are defined as,
\begin{equation}
\begin{split}
\dot{\Pi}+l\Pi=\Phi_f\Phi_f^T;~~\dot{K}+lK=\Phi_f\dot{x}_f^T
\end{split}
\label{filt_mat}
\end{equation}
where, $l>0$ is a scalar and $\Pi : \mathbb{R}^n \to \mathbb{R}^{(k_{w_1}+k_{w_2})\times (k_{w_1}+k_{w_2})}$ and $K: \mathbb{R}^n \to \mathbb{R}^{(k_{w_1}+k_{w_2})\times n}$. 
The solution to (\ref{filt_mat}) with initial conditions $\Pi(0)=0$ and $K(0)=0$ can be obtained as following,
\begin{equation}
\small
\begin{split}
\Pi(t)=\int_0^te^{-l(t-s)}\Phi_f(s)\Phi_f^T(s)ds;~K(t)=\int_0^te^{-l(t-s)}\Phi_f(s)\dot{x}_f(s)^Tds
\end{split}
\label{eq:PK}
\end{equation}
Let $M_1(t)\triangleq\Pi(t)\hat{W}_1-K(t)$. In the subsequent analysis, $M_1 \triangleq M_1(t)$, $M_{1j} \triangleq M_1(t_j)$, $\Pi_j \triangleq \Pi(t_j)$ and $K_j \triangleq K(t_j)$. Then, the update law for identifier NN is given as,
\begin{equation}
\small
\dot{\hat{W}}_1=-\Gamma_1(M_1+\sum_{j=1}^N M_{1j})=-\Gamma_1(\Pi\hat{W}_1-K+\sum_{j=1}^N \Pi_{j}\hat{W}_1-\sum_{j=1}^NK_j)
\label{update}
\end{equation}
Here, $\Gamma_1 \in \mathbb{R}^{(k_{w_1}+k_{w_2})\times(k_{w_1}+k_{w_2})}$ is a positive definite constant learning rate matrix that determines how fast or slow the weights will converge to their true values. 
In (\ref{update}), the term $M_1$ contains the information of error in NN weights, i.e., $\tilde{W}_1=W_1-\hat{W}_1$ as will become clear in the proof of the Theorem \ref{thm:identifier}. Further, the terms under summation represents past values of the term $M_1$ over the memory stack of size $N \geq k_{w_1}+k_{w_2}$. 
This is done to make $\Pi$ a PD matrix, as will become clearer in the proof.


\begin{assumption}\label{as_phi}
\textnormal{The drift dynamics is Lipschitz continuous in $x$ over a compact set $\Omega \in \mathbb{R}^n$ i.e., $\|f(x)\| \leq L_m\|x\|$ and control coupling dynamics is bounded such that $\|g(x)\| \leq g_M$, where $L_m >0$ and $g_M>0$. 
Following from (\ref{approx}) and (\ref{identifier}), as NNs are used to approximate this Lipschitz continuous vector field, the regressor vector, $\Phi$ is bounded such that, there exists, $p_1 \geq 0$, such that, $\|\Phi\| \leq p_1$.
This is a valid assumption because $x \in \Omega \subset \mathbb{R}^n$ and control $u$ is restricted over the set $[-u_m,u_m]$. Further, the combined approximation error $\varepsilon_T$ is bounded such that, there exists, $p_2 \geq 0$, such that, $\|\varepsilon_T\| \leq p_2$. }
\end{assumption}
This is also in line with assumption made in Section 3.1 of \cite{lv2016online}.
Note that due to Assumption \ref{as_phi}, the low-pass filtered versions of $\Phi$ and $\varepsilon_{T}$ are bounded as well.
\begin{definition}\label{PE}
The function $\Phi$ is considered to be persistently excited (PE), if and only if there exist positive scalar constants, $\alpha_1,\alpha_2,T_1$ for all $t>0$ 
such that
\begin{equation}
\alpha_1 I \leq \int_{t}^{t+T_1}\Phi(x(\tau),u(\tau))\Phi^T(x(\tau),u(\tau))d\tau \leq \alpha_2 I
\label{pe1}
\end{equation}
\end{definition}

{Considering that $\Phi \triangleq \Phi(x(\tau),u(\tau))$ is persistently excited, Definition \ref{PE} implies that the matrix $\Pi(t)$ defined in (\ref{eq:PK}) is positive semi-definite (PSD) for time $t<T_1$ and positive definite (PD) for all time $t\geq T_1$.}

\begin{theorem}\label{thm:identifier}
If the regressor vector $\Phi$ is persistently excited (PE) then the update law defined in (\ref{update}) ensures asymptotic stability for $\tilde{W}_1$ i.e., ($\tilde{W}_1=W_1-\hat{W}_1$) when NN identifier approximation error is 0 and UUB stability when NN identifier approximation error is not 0 for the dynamics defined in (\ref{eq:affine1}) $\forall t \geq T_1$ (where, $T_1$ is as defined in definition \ref{PE}).
\end{theorem}
\begin{proof}
From (\ref{identifier}) and (\ref{filt_reg}), Neural Network representation of $\dot{x}_f$ can be obtained as $\dot{x}_f=W_1^T\Phi_f+\varepsilon_{Tf}$, where, $\varepsilon_{Tf}$ is the filtered version of $\varepsilon_T$ defined as, $k\dot{\varepsilon}_{Tf}+\varepsilon_{Tf}=\varepsilon_T$.
Now, utilizing NN representation of $\dot{x}_f$ in $K$ mentioned in (\ref{eq:PK}), it can be written as, 
\begin{equation}
  K=\Pi W_1-y_1  
\label{eq:k}
\end{equation}
where, $y_1$ is given by,
\begin{equation}
    y_1=-\int_0^te^{-l(t-\tau)}\Phi_f(x(\tau),u(\tau))\varepsilon^T_{Tf}(\tau)d\tau
\label{eq:y1}
\end{equation}. 
Using (\ref{eq:k}) in the definition of $M_1$ above,
\begin{equation}
M_1=\Pi\hat{W}_1-K=\Pi\hat{W}_1-\Pi W_1+y_1=-\Pi\tilde{W}_1+y_1
\label{M1}
\end{equation}
Thus, in (\ref{update}), the term $M_1$ contains the information of error in NN weights.
Now, from (\ref{update}) and (\ref{M1}) and using the fact that, $\dot{\tilde{W}}_1=-\dot{\hat{W}}_1$, the dynamics of error in NN weights is expressed as, 
\begin{equation}
\dot{\tilde{W}}=\Gamma_1(-\Pi\tilde{W}_1+y_1-\sum_{j=1}^N\Pi_j\tilde{W}_1+\sum_{j=1}^Ny_{1j})
\end{equation}

Let the Lyapunov candidate be, $V_1=(1/2)tr(\tilde{W}_1^T\Gamma_1^{-1}\tilde{W}_1)$, where $\tilde{W}_1=W_1-\hat{W}_1$.


The time derivative of $V_1$ is given as, 

\begin{equation}
\begin{split}
\dot{V}_1=tr(\tilde{W}_1^T\Gamma_1^{-1}\dot{\tilde{W}}_1)=-tr(\tilde{W}_1^T(\Pi+\sum_{j=1}^N\Pi_j)\tilde{W}_1)+tr(\tilde{W}^T_1(y_1+\sum_{j=1}^Ny_{1j}))
\end{split}
\label{eq:firstvdot}
\end{equation}
In this scenario, i.e., $t \geq T_1$, $\Pi$ matrix will be a PD matrix.

\textbf{Case(i)}: When there is no approximation error, i.e., $y_1=0$

\begin{equation}
\dot{V}_1=-tr(\tilde{W}_1^T(\Pi+\sum_{j=1}^N\Pi_j)\tilde{W}_1) \leq -\lambda_{min}(P)\|\tilde{W}_1\|^2
\label{eq:asymp_stab}
\end{equation}
where, $P\triangleq \Pi+\sum_{j=1}^N\Pi_j$ and $\lambda_{min}$ corresponds to the minimum eigenvalue. 
{Since $P$ is a PD matrix for time $t\geq T_1$ due to persistently excited $\Phi$, $\lambda_{min}(P)>0$ for $t\geq T_1$. 
It is because of this reason that, \eqref{eq:asymp_stab} ensures asymptotic stability of error in identifier NN weights.}

\textbf{Case(ii)}: When there is approximation error, i.e., $y_1 \neq 0$.
\begin{equation}
\begin{split}
\dot{V}_1=tr(\tilde{W}_1^T\Gamma_1^{-1}\dot{\tilde{W}}_1) \leq -\lambda_{min}(P)\|\tilde{W}_1\|^2 + \nu_1\|\tilde{W}_1\|
\end{split}
\label{eq:w1dot}
\end{equation}
where, $\|y_1+\sum_{j=1}^Ny_{1j}\| \leq \nu_1$. 
This follows from using Assumption \ref{as_phi} in (\ref{eq:y1}).

From (\ref{eq:w1dot}) it can be concluded that, $\dot{V}_1$ is negative definite provided,
\begin{equation}
\begin{split}
\|\tilde{W}_1\| > \frac{\nu_1}{\lambda_{min}(P)}
\end{split}
\label{eq:UUB_stab}
\end{equation}
Eq. (\ref{eq:UUB_stab}) implies stability of $\tilde{W}_1$ in the sense of UUB.
\end{proof}


\begin{remark}\label{re1}
There are significant differences between the identifier-critic structure presented in this paper and those in \cite{modares2013adaptive} and \cite{lv2016online}. 
The concept of ER was leveraged in identifier presented in  \cite{modares2013adaptive}, but in this paper error between actual system states and identifier states was minimized instead of directly minimizing error in NN weights, which is more desirable to achieve. 
While \cite{lv2016online} tried to minimize the error in NN weights directly, past observations pertaining to error in NN weights were not effectively utilized by them. 
{Unlike both of these papers, the ER-based identifier NN parameter update law presented in this paper utilizes past observation of errors in NN weights more effectively and attempts to minimize the error in identifier NN weights directly.
Based on the arrival of new data i.e., new $\Pi(t)$ matrix, the second term in matrix $P$ i.e., $\sum_{j=1}^N\Pi_j$ gets updated 
, which in turn implies increasing $\lambda_{min}(P)$ for $t>T_1$ due to the PE consideration of $\Phi$.
This, in effect, yields a tighter UUB bound for $\tilde{W}_1$ in (\ref{eq:UUB_stab}), when there are approximation errors in NN identifier.}
\end{remark}

\section{Optimal tracking control problem and value function approximation}\label{prelim}
\subsection{Preliminaries}
Since the dynamics are assumed to be unknown, their estimated version, particularly the estimated control coupling dynamics will be used in the synthesis of IRL-tracking controller.
The identified dynamics in terms of updated identifier weights are as mentioned in (\ref{eq:fgha}) and (\ref{eq:final_xhat}).

\begin{assumption}\label{uncert}
\textnormal{Following Assumption \ref{as_phi}, the approximated drift dynamics is Lipschitz continuous in $x$ on the compact set $\Omega \subseteq \mathbb{R}^n$ i.e., $\exists L_{1f}>0,\ni \|\hat{f}(x)\| \leq L_{1f}\|x\|,~ \forall x \in \Omega$.  
Similarly, there exists a positive bound ($ g_{M}>0 $) over approximated control coupling dynamics, such that, $~0<\|\hat{g}(x)\|<g_{M}$.
}
\end{assumption}
\begin{assumption}\label{ref_dyn}
\textnormal{Let $x_d(t)$ be the desired reference trajectory which is bounded and governed by $\dot{x}_d(t)=H(x_d(t)) \in \mathbb{R}^n$ and $H(0)=0$, where $H(x_d(t))$ is Lipschitz continuous in $x_d$.}
\end{assumption}
Assumptions \ref{uncert} and \ref{ref_dyn} are also made in \cite{modares2014optimal}, \cite{yang2015robust}.
Dynamics of tracking error $(e\triangleq x-x_d)$ can be written as:
\begin{equation}
\begin{split}
\dot{e}=\dot{x}-\dot{x}_d=\hat{f}(x_d+e)+\hat{g}(x_d+e)u(t)-H(x_d(t))
\end{split}
\label{e_dyn}
\end{equation}
Therefore, the dynamics of augmented system, given as $z=[e^T,x_d^T]^T$, can compactly be written as:
\begin{equation}
\dot{z}=\hat{F}(z)+\hat{G}(z)u
\label{aug}
\end{equation}
where, $u \in \mathbb{R}^m$, $F:\mathbb{R}^{2n} \to \mathbb{R}^{2n}$ and $G: \mathbb{R}^{2n} \to \mathbb{R}^{2n\times m}$ are given by:
\begin{equation}
\begin{split}
\hat{F}(z)=\begin{pmatrix}
\hat{f}(e+x_d)-H(x_d) \\
H(x_d)
\end{pmatrix},~\hat{G}(z)=
\begin{pmatrix}
\hat{g}(e+x_d) \\
0
\end{pmatrix}
\end{split}
\label{eq:aug_system}
\end{equation}

\begin{assumption}\label{as:e}
Following from Assumptions \ref{uncert} and \ref{ref_dyn}, the approximated augmented drift dynamics is Lipschitz continuous in augmented state $z$ such that (s.t), $\|\hat{F}(z)\| \leq L_F\|z\|$ and control coupling dynamics is bounded s.t $\|\hat{G}(z)\| \leq g_M$, where $L_F \geq 0$.
Due to the definition of $z=(e,x_d)^T$, the augmented drift dynamics is continous in $e$ over compact set $\Omega_e \subset \mathbb{R}^n$, such that, $\|\hat{F}(z)\| \leq L_{f1}\|e\|+L_{f2}$, for two positive constants $L_{f1}$ and $L_{f2}$.
\end{assumption}
 Assumption \ref{as:e} is in line with Assumptions made in \cite{modares2014optimal} (refer to Eq. (80) in \cite{modares2014optimal}). 

One of the prime advantages of considering an augmented system, is that, the controller does not require invertibility of control gain matrix, and a single controller comprising of both steady state controller and transient control can be synthesized as indicated in \cite{modares2014optimal} and  \cite{kiumarsi2014reinforcement}. The infinite horizon discounted cost function for (\ref{aug}) is considered as follows \cite{modares2014optimal} : 

\begin{equation}
V(z(t))=\int_t^{\infty} e^{-\gamma(\tau-t)}[Q(z(\tau))+U(u(\tau))]d\tau    
\label{cost}
\end{equation}
where, $Q(z)=z^TQ_1z$ and $Q_1 \in \mathbb{R}^{2n\times 2n}$ is a positive definite matrix given by:
\begin{equation}
\begin{split}
Q_1=\begin{pmatrix}
Q_{n\times n} & 0_{n\times n} \\
0_{n\times n} & 0_{n\times n}
\end{pmatrix}_{2n\times 2n}
\end{split}
\end{equation}
This choice of $Q_1$ leads to $z^TQ_1z=e^TQe$. 
In (\ref{cost}), the term, $[Q(z(\tau))+U(u(\tau))]$ defines the value of utility or penalty per step.
In order to deal with actuator-constraints, the choice of positive definite $U(u)$ was made in line with \cite{abu2005nearly}, \cite{lyashevskiy1996constrained}, \cite{abu2008neurodynamic} and \cite{modares2013adaptive}, as given below,
\begin{equation}
\begin{split}
U(u)=2u_m\int_0^u (\psi^{-1}(\nu/u_m))^TRd\nu = 2u_m\sum_{i=1}^{m}\int_0^{u_i} (\psi^{-1}(\nu_i/u_m))^TR_i d\nu_i
\end{split}
\end{equation}
where, $R \in \mathbb{R}^{m\times m}$ is a positive definite diagonal matrix, ($\psi: \mathbb{R}^m \to \mathbb{R}^m$) is a function possessing following properties -
\\ (i) It is odd and monotonically increasing 
\\ (ii) It is smooth bounded function such that $|\psi(.)| \leq 1$ 

In literature, $\psi$ has been considered as tanh, erf, sigmoid functions. In this paper, $\psi=\tanh{(.)}$ is followed. This also ensures $U(u)$ to remain positive definite (refer to Lemma \ref{cu}). 
The discount factor, $\gamma \geq 0$, determines the value of utility in future. 

Differentiating (\ref{cost}) along the system trajectories and rearranging the terms,
\begin{equation}
\begin{split}
 \nabla{V}(\hat{F}(z)+\hat{G}(z)u)-\gamma V(z)+z^TQ_1z+U(u)=\mathcal{H}(z,u,\nabla{V})=0   
\end{split}
\label{Ha}
\end{equation}
where, $\mathcal{H}(.)$ represents the Hamiltonian.
Let $V^*(z)$ be the optimal cost function satisfying $\mathcal{H}(.)=0$ and is given by,
\begin{equation}
V^*(z)=\min_{u} \int_t^{\infty}e^{-\gamma(\tau-t)}[z^TQ_1z+U(u)]d\tau
\label{opt_cost}
\end{equation}
Let $V^*\triangleq V^*(z)$. Then, $\mathcal{H}(z,u, \nabla{V}^*)=0$ can also be re-written as, 
\begin{equation}
 \nabla{V}^*(\hat{F}(z)+\hat{G}(z)u)-\gamma V^*+z^TQ_1z+U(u)=0
\label{hjb1} 
 \end{equation}


Differentiating (\ref{hjb1}) with respect to $u$, i.e, $\partial \mathcal{H}/ \partial u=0$, closed-form of optimal control action $u^*$ is obtained as \cite{modares2014optimal} and \cite{liu2015reinforcement}.
\begin{equation}
u^*=-u_m\tanh{\Big(R^{-1}\hat{G}(z)^T{\nabla{V}^*}/2u_m\Big)}
\label{opt_u}
\end{equation}
Using (\ref{opt_u}) and Lemma \ref{cu1}, $U(u)$ can be simplified as, 
\begin{equation}
\small
\begin{split}
U(u)=2u_m\int_0^{-u_m\tanh{A(z)}}\tanh^{-1}(\nu/u_m)^TRd\nu=2u_m^2A^T(z)R\tanh{A(z)}+
u_m^2\sum_{i=1}^{m}R_i\log[1-\tanh^2{A_i(z)}]
\end{split}
\label{hjb_int1}
\end{equation}
where, $A\triangleq(1/2u_m)R^{-1}\hat{G}(z)^T\nabla{V}^* \in \mathbb{R}^m$. 
In the subsequent analysis, $\hat{F}\triangleq\hat{F}(z)$ and $\hat{G}\triangleq\hat{G}(z)$.

\subsection{Approximation of value function}
In this subsection, a single NN structure is utilized to approximate the value function. 
Leveraging Weierstrass approximation theorem \cite{abu2005nearly}, any smooth nonlinear mapping can be approximated by selecting a NN with sufficient number of nodes in the hidden layer. 
The value function is assumed to be smooth and hence can be approximated by a NN.

Consider that, there exist an ideal weight vector $W\in \mathbb{R}^{N_1}$ that can approximate the smooth value function as:
\begin{equation}
V(z)=W^T\vartheta(z)+\epsilon(z)
\label{appro}
\end{equation}
where, $\vartheta(z): \mathbb{R}^{2n} \to \mathbb{R}^{N_1}$ ($N_1$ being the number of nodes in hidden layer) is the regressor vector for critic NN and $\varepsilon(z): \mathbb{R}^{2n} \to \mathbb{R}^{N_1}$ is the approximation error. 
In the subsequent analysis, $\vartheta \triangleq \vartheta(z)$.
Then, gradient of $V(z)$ can be expressed as:
\begin{equation}
\nabla{V}=\nabla\vartheta^T W+\nabla\epsilon(z)
\label{vstar}
\end{equation}
Using (\ref{vstar}) in (\ref{opt_u}), the approximate optimal policy is obtained as,
\begin{equation}
\begin{split}
u=-u_m\tanh{\Big(\frac{1}{2u_m}R^{-1}\hat{G}^T \nabla{\vartheta}^TW+\epsilon_{uu}\Big)}
\end{split}
\label{ustar1}
\end{equation}
where, $\small \epsilon_{uu}=(1/2u_m)R^{-1}\hat{G}^T(z)\nabla{\varepsilon}(z)=\small [\varepsilon_{{uu}_{11}},\varepsilon_{{uu}_{12}},...,\varepsilon_{{uu}_{1m}}]^T \in \mathbb{R}^m$.
Using mean value theorem, it can be re-written as ,
\begin{equation}
u=-u_m\tanh{(\tau_1(z))}+\epsilon_{u}
\label{eq:us1}
\end{equation}
where $\tau_{1}(z)=(1/2u_m)R^{-1}\hat{G}^T \nabla{\vartheta}^TW=[\tau_{11},...,\tau_{1m}]^T \in \mathbb{R}^m$ and 
$\epsilon_{u}=-(1/2)((I_m-diag(\tanh^2{(q)}))R^{-1}\hat{G}^T\nabla{\epsilon})$ 
with $q \in \mathbb{R}^m$ and $q_i \in \mathbb{R}$ considered between $\tau_{1i}+\varepsilon_{uui}$ and $\epsilon_{uui}$ i.e., $i^{th}$ element of $\tau_{1}+\varepsilon_{uu}$ and $\epsilon_{uu}$, respectively such that tangent of $\tanh{(q)}$ is equal to the slope of the line joining $\tanh{(\tau_1+\epsilon_{uu})}$ and $\tanh{\varepsilon_{uu}}$.  
For the detailed proof, refer to Lemma \ref{mean_val_lem}. 
In the subsequent analysis, $\tau_1 \triangleq \tau_1(z)$. 
\section{Preliminaries of Integral Reinforcement Learning (IRL)}\label{vargain}
Now, in order to obviate the requirement of prior knowledge about drift dynamics in Bellman equation, basic formulation of IRL would be presented following \cite{modares2014optimal}. 
IRL can be considered as an alternate formulation of Bellman equation that does not require the drift dynamics.
It is obtained by integrating the infinitesimal version of (\ref{cost}) over time interval $[t-T,t]$. 
 \begin{equation}
V(z(t-T))=\int_{t-T}^te^{-\gamma(\tau-t+T)}[Q(z(\tau))+U(u(\tau))]d\tau+e^{-\gamma T}V(z(t))
\label{irl1}
\end{equation}
In order to preserve the equivalence between (\ref{Ha}) and (\ref{irl1}), $T$ must be selected as small as possible \cite{modares2014optimal} and \cite{vrabie2009neural}.

Now, using the approximation from (\ref{appro}) in (\ref{irl1}), the constant approximation error ($\varepsilon_B$) for a given choice of ideal weights $W$, can be expressed as:
\begin{equation}
\begin{split}
\int_{t-T}^te^{-\gamma(\tau-t+T)}[Q(z)+U_1(u)]d\tau+e^{-\gamma T}W^{T}\vartheta(z(t))
-W^{T}\vartheta(z(t-T))\equiv \varepsilon_{B}
\end{split}
\label{ehjb}
\end{equation}
where, $\varepsilon_B=\varepsilon(z(t-T))-e^{-\gamma T}\varepsilon(z(t))$ and $U_1(u)$ is obtained by substituting (\ref{vstar}) and (\ref{ustar1}) in the expression of $U(u)$ obtained in (\ref{hjb_int1}) as follows.
\begin{equation}
\begin{split}
U_1(u)=-W^T\nabla\vartheta \hat{G}u+
u_m^2\sum_{i=1}^{m}R_i\log[1-\tanh^2{(\tau_{1i}(z)+\varepsilon_{uui})}]
\end{split}
\label{hjb_int2}
\end{equation}
Also, 
\begin{equation}
\begin{split}
\int_{t-T}^te^{-\gamma(\tau-t+T)}\dot{\vartheta}d\tau=\int_{t-T}^te^{-\gamma(\tau-t+T)}\nabla{\vartheta}(\hat{F}+\hat{G}u)d\tau=\Delta\vartheta+\gamma\int_{t-T}^te^{-\gamma(\tau-t+T)}\vartheta d\tau
\end{split}
\label{eq:delvarth1}
\end{equation}

Let, $\Delta\vartheta(z(t))\triangleq e^{-\gamma T}\vartheta(z(t))-\vartheta(z(t-T))$, where $\vartheta$ is the regressor vector for critic NN, then (\ref{eq:delvarth1}) can also be written as,

\begin{equation}
\begin{split}
\Delta\vartheta(z(t))=\int_{t-T}^te^{-\gamma(\tau-t+T)}[\nabla{\vartheta}(\hat{F}+\hat{G}u)-\gamma\vartheta]d\tau
\end{split}
\label{delvarth}
\end{equation}
Now, substituting (\ref{delvarth}) and (\ref{hjb_int2}) in (\ref{ehjb}) and upon simplification HJB approximation error 
\begin{equation}
\begin{split}
&\int_{t-T}^te^{-\gamma(\tau-t+T)}[Q(z)+U_1(u)]d\tau+W^T\int_{t-T}^te^{-\gamma(\tau-t+T)}[\nabla{\vartheta}(\hat{F}+\hat{G}u)-\gamma\vartheta]d\tau=\varepsilon_{B}\\
&\int_{t-T}^te^{-\gamma(\tau-t+T)}[z^TQ_1z-\gamma W^T\vartheta+W^T\nabla{\vartheta}\hat{F}+u_m^2\sum_{i=1}^mR_i\log(1-\tanh^2{(\tau_{1i})})]d\tau=\varepsilon_{HJB}
\end{split}
\label{tracking_bell}
\end{equation}
where, $\varepsilon_{HJB}$ is given by 
\begin{equation}
\begin{split}
\varepsilon_{HJB}&=-\int_{t-T}^t(e^{-\gamma(\tau-t+T)}\nabla{\varepsilon}^T\hat{F}+u_m^2\sum_{i=1}^mR_i\log{(1-\tanh^2{(\tau_{1i}+\varepsilon_{uui}}))}\\
&+u_m^2\sum_{i=1}^mR_i\log{(1-\tanh^2{(\tau_{1i})})}-\gamma \varepsilon)d\tau
\end{split}
\end{equation}
Eq. (\ref{tracking_bell}) can also be written as,
\begin{equation}
\begin{split}
W^T\nabla{\vartheta}\hat{F}&=\varepsilon_{HJB}-e^TQe-U_1-W^T\nabla{\vartheta}\hat{G}u+W^T\gamma\vartheta\\
&=\varepsilon_{HJB}-e^TQe-U_1+u_mW^T\nabla{\vartheta}\hat{G}\tanh{(\tau_1)}+W^T\gamma\vartheta\\
\end{split}
\label{eq:trac_bell}
\end{equation}
In the equation above, $z^TQ_1z=e^TQe$.
Since ideal critic NN weights are not known, their estimates will be used instead. This results in approximate value as $\hat{V}(z)=\hat{W}^T\vartheta(z)$, where $\hat{W}$ is the updated weight. 
Thus, approximate optimal control $\hat{u}$ and instantaneous HJB error $\hat{e}$ are then obtained in Eqs. (\ref{approx_u}) and (\ref{hjb_error}), respectively. 
\begin{equation}
\hat{u}=-u_m\tanh{\Big(\frac{1}{2u_m}R^{-1}\hat{G}^T\nabla{\vartheta}^T\hat{W}}\Big)
\label{approx_u}
\end{equation}
\begin{equation}
\begin{split}
 \hat{e}(z(t))=\int_{t-T}^te^{-\gamma(\tau-t+T)}[Q(z)+\hat{U}(\hat{u})]d\tau+\hat{W}^{T}\Delta\vartheta
\end{split}
\label{hjb_error}
\end{equation}
In the subsequent analysis, $\hat{e}\triangleq\hat{e}(z(t))$, $\hat{U}\triangleq\hat{U}(\hat{u})$ i.e., the approximated version of $U(u)$ obtained by substituting $\nabla{\hat{V}}=\nabla\vartheta^T\hat{W}$ as approximated gradient of value function in $A(z)$ in (\ref{hjb_int1}) and given as, (Eq. (\ref{Uhat})) 
\begin{equation}
\begin{split}
\hat{U}=2u_m^2\tau_{2}^T(z)R\tanh{\tau_2(z)}+
u_m^2\sum_{i=1}^{m}R_i\log[1-\tanh^2{\tau_{2i}(z)}]
\end{split}
\label{Uhat}
\end{equation}
where, $\tau_2(z)=(1/2u_m)R^{-1}\hat{G}^T \nabla{\vartheta}^T\hat{W} \in \mathbb{R}^m$. In the subsequent analysis, $\tau_2\triangleq\tau_2(z)$. 
Now using (\ref{delvarth}) and subtracting (\ref{tracking_bell}) from (\ref{hjb_error}), the instantaneous HJB error (from Eq. (\ref{hjb_error})) can be approximated in terms of $\tilde{W}=W-\hat{W}$, as described in \cite{modares2014optimal}.
\begin{equation}
\begin{split}
\hat{e}(\tilde{W})=-\Delta\vartheta^T\tilde{W}+\int_{t-T}^te^{-\gamma(\tau-t+T)}\tilde{W}^TMd\tau+E
\end{split}
\label{eq:ewtil}
\end{equation}

\begin{equation}
\begin{split}
M&=\nabla{\vartheta}\hat{G}u_m(\tanh{(\tau_2)}-sgn(\tau_2))~\\
E&=\int_{t-T}^te^{-\gamma(\tau-t+T)}\Big(W^T\nabla{\vartheta}\hat{G}u_m(sgn(\tau_{2})-sgn(\tau_{1}))+u_m^2R(\varepsilon_{\tau_2}-\varepsilon_{\tau_1})+\varepsilon_{HJB}\Big)d\tau
\end{split}
\label{em}
\end{equation}
where, $\varepsilon_{\tau_1}$ and $\varepsilon_{\tau_2}$ are bounded approximation errors.
\section{Online variable gain gradient descent-based update law and its stability proof}\label{paramupd}


\subsection{Variable gain gradient descent-based critic update law in On-policy IRL}
Unlike traditional NN weight update laws for IRL algorithm based on gradient descent with constant learning rate such as \cite{vamvoudakis2014online} and \cite{modares2014optimal}, a novel online parameter update law that is driven by variable gain gradient descent is presented in this paper in order to update the critic NN weights in IRL framework as below.  

\begin{equation}
\begin{split}
\dot{\hat{W}}&=-\alpha(|\hat{e}|^{k_2}+l)\bar{\vartheta}\hat{e}+\alpha\Xi(z,\hat{u})\left(\frac{\nabla{\vartheta}^T}{2}\hat{G}R^{-T}(I_m-\mathcal{B})\hat{G}^T\nabla{\vartheta}\hat{W}-\nabla{\vartheta}\dot{z}\right)\\
&+\alpha(|\hat{e}|^{k_2}+l)\left((K_1\varphi^T-K_2)\hat{W}-\bar{\vartheta}\int_{t-T}^te^{-\gamma(\tau-t+T)}\hat{W}^TMd\tau \right)
\end{split}
\label{tuning_law1}
\end{equation}
Here, $0<l\leq 1$, $\alpha>0,~k_2>0$, and 
$K_1 \in \mathbb{R}^{N_1}~ K_2 \in \mathbb{R}^{N_1\times N_1}$ are constants. 
Other terms used in the update law above are, $\bar{\vartheta}=\Delta\vartheta/m_s^2$,   $\varphi=\Delta\vartheta/m_s$, $m_s=1+\Delta\vartheta^T\Delta\vartheta$, $\mathcal{B}=diag\{\tanh^2{(\tau_{2i}(z))}\} \in \mathbb{R}^{m\times m},~i=1,2...,m$, and $M$ is as in (\ref{em}).
A piece-wise continuous indicator function is defined as,
\vspace{-.1cm}
\begin{equation}
\Xi(z,\hat{u})=  \begin{cases} 
      0, & if~ \Sigma <0  \\
      1, & otherwise
   \end{cases}
\end{equation}
Where, $\Sigma$ is the rate of variation of Lyapunov function ($L=V+(1/2)tr(\tilde{W}_1^T\Gamma_1^{-1}\tilde{W}_1)+(1/(2\alpha))\tilde{W}^T\tilde{W}$ (See Eq. (\ref{eq:final_L}), where $V$ is the value function and $\tilde{W}_1$ and $\tilde{W}$ are the error in NN weights of identifier and critic, respectively) along the augmented system trajectories and is denoted by,
\begin{equation}
\Sigma=\dot{L}=\frac{\partial L}{\partial z}\dot{z}=\frac{\partial {V}}{\partial z}\dot{z}=\hat{W}^T\nabla{\vartheta}\dot{z}
\end{equation}
Note that since $\Sigma$ could be computed by numerical differentiation in (\ref{tuning_law1}), explicit knowledge of augmented drift dynamics ($\hat{F}$) is not needed in the update law.
Additionally, the control coupling dynamics i.e., $\hat{G}$ required in (\ref{tuning_law1}) and (\ref{approx_u}) is generated by the online NN identifier introduced in the previous section. 
{In order to facilitate the analytical treatment of the stability analysis, $k_2$ in (\ref{tuning_law1}) will be set to $1$ i.e., $k_2=1$.}

Unlike actor-critic dual approximation structure presented in \cite{vamvoudakis2014online} and \cite{modares2014optimal}, which also requires an initial stabilizing controller, this paper leverages an identifier-critic structure with a stabilizing term (following from such a stabilizing term presented in \cite{dierks2010optimal} for RL formulation) in the critic update law, thus obviating the requirement of initial stabilizing controller in IRL algorithm.


The novelty of the critic scheme presented in this paper i.e., (\ref{tuning_law1}) is explained below.
\begin{itemize}
\item Since IRL framework is utilized, it is inherently independent of prior knowledge about drift dynamics and control coupling dynamics present in (\ref{tuning_law1}) is fed by ER-enhanced identifier. 
Thus, it is different from \cite{mishra2019criticonly}, wherein explicit knowledge of control coupling dynamics was used.
\item The learning rate in critic update law is a variable and can scale the learning rate based on instantaneous value of the HJB error. The first term in (\ref{tuning_law1}) originates from gradient descent (GD) and is responsible for minimizing the instantaneous HJB error.    
The learning rate of this term is $\alpha(|\hat{e}|^{k_2}+l)$, which is a function of approximated value of instantaneous HJB error.
This facilitates accelerating the learning speed when the HJB error is large and dampening the learning speed as the HJB error becomes smaller.

\item The second term in (\ref{tuning_law1}) is responsible for stability.
An indicator function as proposed in \cite{dierks2010optimal} and \cite{liu2015reinforcement} is utilized here so that the second term becomes active only when the Lyapunov function is non-decreasing along the augmented system trajectories. However, the expression of this term is different from
\cite{dierks2010optimal} or \cite{liu2015reinforcement}, and is explained below,
Let $\Sigma<0$ denote the case when $V$ (refer to Eq. (\ref{cost})) is strictly decreasing along the augmented system trajectory, hence if $\Sigma>0$, then the cost $V$ (see Eq. (\ref{cost})) is growing along the augmented system trajectory in closed loop system.
When $\Sigma \geq0$, then gradient descent can be leveraged to ensure decreasing $\Sigma$ over $\hat{W}$, i.e.,
\begin{equation}
\small
\begin{split}
-\alpha\frac{\partial \Sigma}{\partial \hat{W}}&=-\alpha\frac{\partial (\hat{W}^T\nabla{\vartheta}\dot{z})}{\partial \hat{W}}=-\alpha\frac{\partial (\hat{W}^T\nabla{\vartheta}(\hat{F}+\hat{G}\hat{u}))}{\partial \hat{W}}\\
&=\alpha\left(\frac{\nabla{\vartheta}^T}{2}\hat{G}R^{-1}(I_m-\mathcal{B})\hat{G}^T\nabla{\vartheta}\hat{W}-\nabla{\vartheta}\dot{z}\right)
\end{split}
\end{equation}
\item The last term in (\ref{tuning_law1}) improves the robustness of the update law against the variation in function approximation errors etc and help in controlling the size of the residual sets for error in critic NN as will become clear in the stability proof of Theorem \ref{thm:crit} and the remarks its subsequent discussion.
\end{itemize}

The error in critic weight is defined as, $\tilde{W}=W-\hat{W}$. Utilizing (\ref{eq:ewtil}), the critic weight error dynamics is then given as, 
\begin{equation}
\begin{split}
\dot{\tilde{W}}&=\alpha g_1(z(t))\bar{\vartheta}[-\Delta\vartheta^T\tilde{W}+\int_{t-T}^te^{-\gamma(\tau-t+T)}\tilde{W}Md\tau+E]\\
& -\alpha\Xi(z,\hat{u})\left(\frac{\nabla{\vartheta}^T}{2}\hat{G}R^{-1}(I_m-\mathcal{B})\hat{G}^T\nabla{\vartheta}\hat{W}-\nabla{\vartheta}\dot{z}\right) \\
&+\alpha g_1(z(t))\Big[\bar{\vartheta}\int_{t-T}^te^{-\gamma(\tau-t+T)}\hat{W}Md\tau+(K_2-K_1\varphi^T)\hat{W}\Big]
\end{split}
\label{crit_err}
\end{equation}
where, $g_1(z(t))\triangleq|\hat{e}(z(t))|+l$, where $0<l\leq 1$ is a small positive constant and $\hat{e}$ is instantaneous HJB error. 
In the subsequent analysis $g_1 \triangleq g_1(z(t))$.
\vspace{-.5cm}
\subsection{Stability proof of the online update law}
In the critic update law (See Eq. (\ref{tuning_law1})) presented in this paper, there exists a stabilizing term (second term) in the update law (\ref{tuning_law1}) that comes into effect when Lyapunov function starts growing along the system trajectories. 
This term helps in pulling the system out of instability ensuring that the augmented system trajectories remain finite when starting from an initial point $z_0$ within a compact set $\Omega \subset R^{2n}$.
This fact leads to the following two assumptions.
\begin{assumption}\label{weight_as}
\textnormal{Ideal NN weight vector $W$ is considered to be bounded, i.e., $\|W\| \leq  W_M$. There also exist positive constants $b_{\epsilon}$ and $b_{\epsilon z}$ that bound the approximation error and its gradient such that $\|\varepsilon(z)\| \leq b_{\epsilon}$ and $\|\nabla{\varepsilon}\| \leq b_{\epsilon z}$. }
\end{assumption}
This is in line with Assumptions 3b of \cite{vamvoudakis2010online},  Assumption 5 of \cite{liu2015reinforcement} and Assumptions made in Section 4.1 in \cite{modares2014optimal}
\begin{assumption}\label{regres_as}
\textnormal{Critic regressors are considered to be bounded as well: $\|\vartheta(z)\| \leq b_{\vartheta}$ and $\|\nabla{\vartheta}(z)\| \leq b_{\vartheta z}$. }
\end{assumption}
This is in line with Assumption 4 of \cite{yang2015robust}, Assumption 6 of \cite{liu2015reinforcement} and Assumption 4 of \cite{modares2014optimal}


\begin{theorem}\label{thm:crit}
Let the CT nonlinear augmented system be described by (\ref{aug}) with associated Bellman equation as (\ref{irl1}) and approximate optimal control as (\ref{approx_u}), then, under the assumptions \ref{uncert}-\ref{regres_as} and if the regressor $\Delta\vartheta$ is persistently excited (PE), then the tuning laws given by (\ref{update}) for identifier and (\ref{tuning_law1}) for critic NN guarantee that the quantities $e$, $\tilde{W}_1$ and $\tilde{W}$ are uniform ultimate bounded (UUB), where, $e$ is the state error and $\tilde{W}_1$, $\tilde{W}$ are the error in NN weights of identifier and critic NNs, respectively.
\end{theorem}

\begin{proof}
Let the Lyapunov candidate be
\begin{equation}
  L=V+\frac{1}{2}tr(\tilde{W}_1^T\Gamma_1^{-1}\tilde{W}_1)+\frac{1}{2\alpha}\tilde{W}^T\tilde{W}  
\label{eq:final_L}
\end{equation}
where, $\Gamma_1 >0$ is the learning rate for the identifier NN as defined after (\ref{update}) and $\alpha>0$ is the constant baseline learning rate of the critic NN.  

Then, time derivative of Lyapunov function becomes,
\begin{equation}
\small
\begin{split}
\dot{L}
&=\dot{V}+tr(\tilde{W}_1^T\Gamma_1^{-1}\dot{\tilde{W}}_1)+\tilde{W}^T\dot{\tilde{W}}/\alpha = \underbrace{W^T\nabla{\vartheta}(\hat{F}-u_m\hat{G}\tanh{(\tau_2)})+\varepsilon_0}_\text{= $\dot{V}$}+tr(\tilde{W}_1^T\Gamma_1^{-1}\dot{\tilde{W}}_1)+\tilde{W}^T\dot{\tilde{W}}/\alpha
\end{split}
\label{eq:Ldot}
\end{equation}
In (\ref{eq:Ldot}), $\dot{V}=\nabla{V}\dot{z}$ and hence, $\dot{V}=\nabla{\vartheta}^TW\dot{z}+\underbrace{\nabla{\varepsilon}\dot{z}}_\text{$\varepsilon_{0}$}$, therefore, $\varepsilon_0=\nabla{\varepsilon}(\hat{F}-u_m\hat{G}\tanh{(\tau_2)})$.
Utilizing the expression of $W^T\nabla{\vartheta}\hat{F}$ from (\ref{eq:trac_bell}) in (\ref{eq:Ldot}),
\begin{equation}
\small
\begin{split}
\dot{L}&=\underbrace{-e^TQe-U_1+\gamma W^T\vartheta+u_mW^T\nabla{\vartheta}\hat{G}(\tanh{(\tau_1)}-\tanh{(\tau_2)})+\varepsilon_{HJB}+\varepsilon_0}_\text{= $\dot{V}$}+tr(\tilde{W}_1^T\Gamma_1^{-1}\dot{\tilde{W}}_1)+\tilde{W}^T\dot{\tilde{W}}/\alpha
\end{split}
\end{equation}
 
Now, using Assumptions \ref{as:e} and \ref{weight_as} the bound over it can be written as, 
\begin{equation}
\begin{split}
\|\varepsilon_0\| \leq b_{\varepsilon z}(L_{f1}\|e\|+L_{f2}+g_Mu_m)
\end{split}
\label{lfe}
\end{equation}
For some positive constant $L_{f1}$ and $L_{f2}$.
Utilizing assumptions \ref{as:e}, \ref{weight_as}, \ref{regres_as}, Eq. (\ref{lfe}), Lemma \ref{tanhlem}, the bound over $\dot{V}$ can be expressed as,
\begin{equation}
\begin{split}
\dot{V} \leq -\lambda_{min}(Q)\|e\|^2+\gamma W_Mb_{\vartheta}+b_{\varepsilon z}(L_{f1}\|e\|+L_{f2}+g_Mu_m)+\varepsilon_h+u_mW_Mb_{\vartheta z}g_MT_m
\end{split}
\label{vdt}
\end{equation}
From assumptions \ref{as:e} and \ref{weight_as}, $\varepsilon_{HJB}$ can be assumed to be bounded as $|\varepsilon_{HJB}|\leq \varepsilon_h$ 
Using error dynamics of weights (\ref{crit_err}), the last term of $\dot{L}$, i.e., $\tilde{W}^T\dot{\tilde{W}}/\alpha$ becomes: 
\begin{equation}
\small
\begin{split}
\tilde{W}^T&\alpha^{-1}\dot{\tilde{W}}
=g_1\tilde{W}^T\bar{\vartheta}\Big[-\Delta\vartheta^T\tilde{W}+\int_{t-T}^te^{-\gamma(\tau-t+T)}\tilde{W}Md\tau+E\Big]\\
&-\tilde{W}^T\Xi(z,\hat{u})\left(\frac{\nabla{\vartheta}^T}{2}\hat{G}R^{-1}(I_m-\mathcal{B})\hat{G}^T\nabla{\vartheta}\hat{W}-\nabla{\vartheta}\dot{z}\right)\\ 
&+g_1\tilde{W}^T\bar{\vartheta}\int_{t-T}^te^{-\gamma(\tau-t+T)}\hat{W}Md\tau+\underbrace{g_1\tilde{W}^T(K_2\hat{W}-K_1\varphi^T\hat{W})}_\text{$\triangleq c_1$} \\
&=\underbrace{-g_1\tilde{W}^T\varphi\varphi^T\tilde{W}+g_1\tilde{W}^T(\varphi^T/m_s)E+g_1\tilde{W}^T\beta_1(z)+g_1\tilde{W}^T(K_2\hat{W}-K_1\varphi^T\hat{W})}_\text{$\triangleq A$}\\
&\underbrace{-\tilde{W}^T\Xi(z,\hat{u})\left(\frac{\nabla{\vartheta}^T}{2}\hat{G}R^{-1}(I_m-\mathcal{B})\hat{G}^T\nabla{\vartheta}{W}\right)}_\text{$\triangleq S$}\\
&\underbrace{+\tilde{W}^T\Xi(z,\hat{u})\left(\frac{\nabla{\vartheta}^T}{2}\hat{G}R^{-1}(I_m-\mathcal{B})\hat{G}^T\nabla{\vartheta}\tilde{W}\right)-\Xi(z,\hat{u})\tilde{W}^T\nabla{\vartheta}\dot{z}}_\text{$\triangleq S$}
\end{split}
\label{ldot4}
\end{equation}
where, $\Delta\vartheta$ is as defined after (\ref{eq:delvarth1}), $\bar{\vartheta}=\varphi/m_s$, $\varphi=\Delta\vartheta/m_s$,  $\beta_1(z)=(\varphi^T/m_s)I_0$ where $I_0=\int_{t-T}^te^{-\gamma(\tau-t+T)}W^TMd\tau$ and $M,E$ as defined in (\ref{em}) and $m_s$ as defined after (\ref{tuning_law1}).
The term $S$ in (\ref{ldot4}) is the stabilizing term written as a function of $W$ and $\tilde{W}$ i.e.,
\begin{equation}
\begin{split}
S&=-\tilde{W}^T\Xi(z,\hat{u})\left(\frac{\nabla{\vartheta}^T}{2}\hat{G}R^{-1}(I_m-\mathcal{B})\hat{G}^T\nabla{\vartheta}{W}\right)\\
&+\tilde{W}^T\Xi(z,\hat{u})\left(\frac{\nabla{\vartheta}^T}{2}\hat{G}R^{-1}(I_m-\mathcal{B})\hat{G}^T\nabla{\vartheta}\tilde{W}\right)-\Xi(z,\hat{u})\tilde{W}^T\nabla{\vartheta}\dot{z}
\end{split}
\label{eq:St}
\end{equation}
The term $c_1$ in (\ref{ldot4}) can be expressed as:
\begin{equation}
\begin{split}
 \tilde{W}^T(K_2\hat{W}-K_1\varphi^T\hat{W})=\tilde{W}^TK_2W-\tilde{W}^TK_2\tilde{W}-\tilde{W}^TK_1\varphi^TW+\tilde{W}^TK_1\varphi^T\tilde{W}
\end{split}
\label{eq:K}
\end{equation}
Now, Let us define,  
\begin{equation}
\begin{split}
  \mathcal{S}\triangleq[\tilde{W}^T\varphi,\tilde{W}^T]^T  
\end{split}
\label{j_def}
\end{equation}
Then using (\ref{eq:K}) and (\ref{j_def}), the term $A$ in (\ref{ldot4}) can be compactly written as, $A=g_1(-\mathcal{S}^TM_1\mathcal{S}+\mathcal{S}^TN)$.
Therefore, $\tilde{W}^T\dot{\tilde{W}}/\alpha$ can be written as,
\begin{equation}
\begin{split}
\dot{\tilde{W}}^T\tilde{W}/\alpha=A+S\leq g_1(-\lambda_{min}(M_1)\|\mathcal{S}\|^2+b_N\|\mathcal{S}\|)+S
\end{split}
\label{eq:wtildedot}
\end{equation}
where, $N\in \mathbb{R}^{N_1+1}$ and $M_1\in \mathbb{R}^{(N_1+1)\times (N_1+1)}$ are defined as,
\begin{equation}
\small
\begin{split}
M_1=\begin{pmatrix}
I    &      -\frac{1}{2}K^T_1 \\
-\frac{1}{2}K_1  &  K_2
\end{pmatrix};~~N=\begin{pmatrix}
E/m_s  \\
(\beta_1(z)+K_2W_c-K_1\varphi^TW)
\end{pmatrix}
\end{split}
\label{eq:MN}
\end{equation}
where, $b_N$ is the upper bound over vector $N$ i.e., 
\begin{equation}
\|N\| \leq b_N=\max(\|N\|)
\label{bn}
\end{equation}
where, $N_1$ is the size of the regressor vector for critic NN.
In (\ref{eq:MN}), if $K_1$ and $K_2$ are chosen such that $K_2$ is symmetric, then $M_1$ becomes symmetric.
Further, in order to ensure that $\lambda_{min}(M_1)$ is real and positive, $K_1$ and $K_2$ should be selected such that $M_1$ is positive definite.
Further, $A$ can be developed by leveraging $g_1$ as a function of $\tilde{W}$.
{From (\ref{eq:ewtil}), $g_1$ as a function of $\tilde{W}$ is, 
$g_1=|\hat{e}(\tilde{W})|^{k_2}+l$. 
}
From (\ref{eq:ewtil}) and using small $T$, 
\begin{equation}
\begin{split}
g_1&=|-\Delta\vartheta^T\tilde{W}+\int_{t-T}^te^{-\gamma(\tau-t+T)}\tilde{W}^TMd\tau+E|+l\\
&\leq \|E\|+\|\tilde{W}\|\|\Delta\vartheta\|+\|Te^{-\gamma T}\tilde{W}M\|+l
\end{split}
\end{equation}
Finally, utilizing $e^{-\gamma T}\approx1-\gamma T$ and $\|\tanh{(\tau_2)}-sgn(\tau_2)\| \leq 2\sqrt{m}$, where $\tau_2 \in \mathbb{R}^m$.
\begin{equation}
\begin{split}
g_1&\leq \|\tilde{W}\|(\|\Delta\vartheta\|+\|Te^{-\gamma T}M\|)+\|E\|+l\\
&\leq\|\tilde{W}\|(\|\Delta\vartheta\|+\|Tb_{\vartheta z}g_Mu_m2\sqrt{m}\|)+\|E\|+l
\end{split}
\end{equation}
Utilizing Assumption \ref{regres_as} and the definition of $\Delta\vartheta$ after (\ref{eq:delvarth1}), the bound over $g_1$ can be written as,
\begin{equation}
g_1\leq \|\tilde{W}\|\underbrace{(\|\gamma Tb_{\vartheta}\|+\|Tb_{\vartheta z}g_Mu_m2\sqrt{m}\|)}_\text{$\triangleq A_1$}+\underbrace{\|E\|+l}_\text{$\triangleq A_2$}
\end{equation}
It could be noted that $E$ is one of the component of vector $N$, and by appropriately selecting $K_1$ and $K_2$ in $N$, and selecting a very small offset $l$, it can be ensured that $A_2=\|E\|+l \leq b_N$. 
Also, from (\ref{j_def}), $\|\tilde{W}\| \leq \|\mathcal{S}\|$ 
, therefore,
\begin{equation}
\begin{split}
g_1\leq A_1\|\mathcal{S}\|+b_N
\end{split}
\label{eq:g1}
\end{equation}
 
Now, using (\ref{eq:w1dot}), (\ref{vdt}), (\ref{eq:wtildedot}) and (\ref{eq:g1}) in (\ref{eq:Ldot}) the Lyapunov derivative can be rendered as,
\begin{equation}
\begin{split}
\dot{L} &\leq -\lambda_{min}(Q)\|e\|^2+\gamma W_Mb_{\vartheta}+b_{\varepsilon z}(L_{f1}\|e\|+L_{f2}+g_Mu_m)+\varepsilon_h+u_mW_Mb_{\vartheta z}g_MT_m\\
&-\lambda_{min}(P)\|\tilde{W}_1\|^2+\nu_1\|\tilde{W}_1\|+(A_1\|\mathcal{S}\|+b_N)(-\lambda_{min}(M_{1})\|\mathcal{S}\|^2+b_N\|\mathcal{S}\|)+S
\end{split}
\label{ldot}
\end{equation}
Based on the rate of variation of Lyapunov function along the system trajectories, 
Eq. (\ref{ldot}) can be analyzed for the following two cases as detailed below.

\textbf{Case (i)}:
$\Xi(z,\hat{u})=0$.
\\
This case implies that the Lyapunov function is strictly decreasing along the augmented system trajectories, and hence the stabilizing term, i.e., $S$ in (\ref{ldot}) is $0$. 
\begin{equation}
\small
\begin{split}
\dot{L}
&\leq -\lambda_{min}(Q)\|e\|^2+a_1\|e\|+\underbrace{\gamma W_Mb_{\vartheta}+b_{\varepsilon z}L_{f2}+b_{\varepsilon z}g_Mu_m+\varepsilon_h+u_mW_Mb_{\vartheta z}g_MT_m}_\text{$\triangleq a_2$}\\
&-\lambda_{min}(P)\|\tilde{W}_1\|^2+\nu_1\|\tilde{W}_1\|-\underbrace{\lambda_{min}(M_1)A_1\|\mathcal{S}\|^{3}+b_NA_1\|\mathcal{S}\|^{2}-\lambda_{min}(M)b_N\|\mathcal{S}\|^2+b_N^2\|\mathcal{S}\|}_\text{$A$}\\
&\leq -\lambda_{min}(Q)\left(\|e\|-\frac{a_1}{2\lambda_{min}(Q)}\right)^2+\underbrace{\frac{a_1^2}{4\lambda_{min}(Q)}+a_2+\frac{\nu_1^2}{4\lambda_{min}(P)}}_\text{$\triangleq \aleph$}-\lambda_{min}(P)\left(\|\tilde{W}_1\|-\frac{\nu_1}{2\lambda_{min}(P)}\right)^2\\
&-\underbrace{\lambda_{min}(M_1)A_1\|\mathcal{S}\|^{3}+b_NA_1\|\mathcal{S}\|^{2}-\lambda_{min}(M_1)b_N\|\mathcal{S}\|^2+b_N^2\|\mathcal{S}\|}_\text{$A$}
\end{split}
\label{eq:ldot11}
\end{equation}
where, $a_1\triangleq b_{\varepsilon z}L_{f1}$. 
In order for $\dot{L}$ to be negative definite, following inequalities should hold true,
\begin{equation}
\begin{split}
 -\lambda_{min}&(Q)\left(\|e\|-\frac{a_1}{2\lambda_{min}(Q)}\right)^2+\aleph<0\\
\Rightarrow&\|e\| >\frac{a_1}{2\lambda_{min}(Q)}+\sqrt{\frac{\aleph}{\lambda_{min}(Q)}}
\end{split}
\label{declyap}
\end{equation}
or
\begin{equation}
\begin{split}
-\lambda_{min}&(P)\left(\|\tilde{W}_1\|-\frac{\nu_1}{2\lambda_{min}(P)}\right)^2+\aleph<0\\   
 \Rightarrow&\|\tilde{W}_1\| > \frac{\nu_1}{2\lambda_{min}(P)}+\sqrt{\frac{\aleph}{\lambda_{min}(P)}}
\end{split}
\label{eq:w1uub}
\end{equation}
and 
\begin{equation}
\begin{split}
-\lambda_{min}(M_1)A_1\|\mathcal{S}\|^{3}+b_NA_1\|\mathcal{S}\|^{2}-\lambda_{min}(M_1)b_N\|\mathcal{S}\|^2+b_N^2\|\mathcal{S}\|<0
\end{split}
\label{script_s}
\end{equation}
For $\|\mathcal{S}\|\neq 0$, Eq. (\ref{script_s}) implies,
\begin{equation}
\begin{split}
-&A_1\lambda_{min}(M_1)\|\mathcal{S}\|^2+\underbrace{(b_NA_1-\lambda_{min}(M_1)b_N)}_\text{$\triangleq B_1$}\|\mathcal{S}\|+b_N^2<0\\
\Rightarrow&\|\mathcal{S}\|>\frac{B_1}{2A_1\lambda_{min}(M_1)}+\sqrt{\frac{B_1^2}{4A_1^2\lambda^2_{min}(M_1)}+\frac{b_N^2}{A_1\lambda_{min}(M_1)}}\\
\Rightarrow&\|\mathcal{S}\|>\frac{b_N}{\lambda_{min}(M_1)}\underbrace{\left[\frac{1}{2}\left(1-\frac{\lambda_{min}(M_1)}{A_1}\right)+\sqrt{\frac{1}{4}\left(1-\frac{\lambda_{min}(M_1)}{A_1}\right)^2-\frac{\lambda_{min}(M_1)}{A_1}}\right]}_\text{$\triangleq \Gamma$}
\end{split}
\label{eq:S}
\end{equation}
Let, $\gamma_1 \triangleq \frac{\lambda_{min}(M_1)}{A_1}$, therefore, if, $0 \leq \gamma_1 \leq 3-\sqrt{8}\approx0.17$, then $.478\leq \Gamma \leq 1$.
Also, recall from the definition of $\mathcal{S}$ in (\ref{j_def}), the upper bound of $\|\mathcal{S}\|$ can be obtained as,
\begin{equation}
 \|\mathcal{S}\| \leq \Big(\sqrt{1+\|\varphi\|^2}\Big)\|\tilde{W}\|
 \label{Y-ineq}
\end{equation}
Therefore, from lower and upper bounds of $\mathcal{S}$ in (\ref{eq:S}) and (\ref{Y-ineq}), respectively, the bound over $\|\tilde{W}\|$ becomes,  
\begin{equation}
\|\tilde{W}\| > \frac{\frac{b_N}{\lambda_{min}(M_1)}\Gamma}{\sqrt{1+\|\varphi\|^2}} 
\label{ld1}
\end{equation}

It could be seen that $\tilde{W}$, $\tilde{W}_1$ and $e$ are UUB stable with corresponding bounds given in the RHS of (\ref{ld1}), (\ref{eq:w1uub}) and (\ref{declyap}), respectively. 
Also, note that if Eq. (\ref{ld1}) along with either of the Eqs. (\ref{eq:w1uub}) or (\ref{declyap}) holds, then the negative definiteness of $\dot{L}$ is ensured.

\textbf{Case (ii)}: $\Xi(z,\hat{u})=1$ \\
This case implies that the Lyapunov function is non-decreasing along the augmented system trajectories, 
that is $\Sigma\geq 0$.

From (\ref{eq:St}), the bound over $S$ is, 
\begin{equation}
\begin{split}
\|S\| \leq  S_M=\|\tilde{W}\|\frac{1}{2}b^2_{\vartheta z}W_Mg_M^2\|R^{-1}\|+\|\tilde{W}\|^2\frac{1}{2}b^2_{\vartheta z}g_M^2\|R^{-1}\|+\|\tilde{W}\|b_{\vartheta z}(L_f\|z\|+g_Mu_m)
\end{split}
\end{equation}
Utilizing Assumption \ref{as:e},
\begin{equation}
\small
\begin{split}
\|S\| \leq  S_M=\|\tilde{W}\|\frac{1}{2}b^2_{\vartheta z}W_Mg_M^2\|R^{-1}\|+\|\tilde{W}\|^2\frac{1}{2}b^2_{\vartheta z}g_M^2\|R^{-1}\|+\|\tilde{W}\|b_{\vartheta z}(L_{f1}\|e\|+L_{f2}+g_Mu_m)
\end{split}
\label{eq:sbound}
\end{equation}
Using (\ref{eq:sbound}) and the first inequality of (\ref{eq:ldot11}) in (\ref{ldot}),
\begin{equation}
\begin{split}
\dot{L}&\leq-\lambda_{min}(Q)\|e\|^2+a_1\|e\|+a_2-\lambda_{min}(P)\|\tilde{W}_1\|^2+\nu_1\|\tilde{W}_1\|\\
&-\lambda_{min}(M_1)A_1\|\mathcal{S}\|^{3}+b_NA_1\|\mathcal{S}\|^{2}-\lambda_{min}(M_1)b_N\|\mathcal{S}\|^2+b_N^2\|\mathcal{S}\|\\
&+\|\tilde{W}\|\frac{1}{2}b^2_{\vartheta z}W_Mg_M^2\|R^{-1}\|+\|\tilde{W}\|^2\frac{1}{2}b^2_{\vartheta z}g_M^2\|R^{-1}\|+\|\tilde{W}\|b_{\vartheta z}(L_{f1}\|e\|+L_{f2}+g_Mu_m)
\end{split}
\label{eq:first_ldot}
\end{equation}

Now, two numbers $l_1$ and $l_2$ are considered such that $l_1+l_2=1$. Then,
\begin{equation}
\small
\begin{split}
\dot{L}&\leq -l_1\lambda_{min}(Q)\|e\|^2+a_1\|e\|+a_2-\lambda_{min}(P)\|\tilde{W}_1\|^2+\nu_1\|\tilde{W}_1\|-\lambda_{min}(M_1)A_1\|\mathcal{S}\|^{3}\\
&+b_NA_1\|\mathcal{S}\|^{2}-\lambda_{min}(M_1)b_N\|\mathcal{S}\|^2+b_N^2\|\mathcal{S}\|+\|\tilde{W}\|\frac{1}{2}b^2_{\vartheta z}W_Mg_M^2\|R^{-1}\|+\|\tilde{W}\|^2\frac{1}{2}b^2_{\vartheta z}g_M^2\|R^{-1}\|\\
&-l_2\lambda_{min}(Q)\left(\|e\|-\frac{\|\tilde{W}\|b_{\vartheta z}L_{f1}}{2l_2\lambda_{min}(Q)}\right)^2+\frac{(\|\tilde{W}\|b_{\vartheta z}L_{f1})^2}{4l_2\lambda_{min}(Q)}+\|\tilde{W}\|b_{\vartheta z}(L_{f2}+g_Mu_m)
\end{split}
\label{ld11}
\end{equation}

From (\ref{j_def}), $\|\tilde{W}\| \leq \|\mathcal{S}\|$. Therefore, (\ref{ld11}) can be further simplified into,
\begin{equation}
\begin{split}
\dot{L}&\leq -l_1\lambda_{min}(Q)\|e\|^2+a_1\|e\|+a_2-\lambda_{min}(P)\|\tilde{W}_1\|^2+\nu_1\|\tilde{W}_1\|\\
&-\lambda_{min}(M_1)A_1\|\mathcal{S}\|^{3}+b_NA_1\|\mathcal{S}\|^{2}-\lambda_{min}(M_1)b_N\|\mathcal{S}\|^2+b_N^2\|\mathcal{S}\|\\
&+\|\mathcal{S}\|\frac{1}{2}b^2_{\vartheta z}W_Mg_M^2\|R^{-1}\|+\|\mathcal{S}\|^2\frac{1}{2}b^2_{\vartheta z}g_M^2\|R^{-1}\|+\frac{\|\mathcal{S}\|^2b^2_{\vartheta z}L^2_{f1}}{4l_2\lambda_{min}(Q)}+\|\mathcal{S}\|b_{\vartheta z}(L_{f2}+g_Mu_m)
\end{split}
\end{equation}

Upon further simplification, the above inequality can be rendered as,
\begin{equation}
\small
\begin{split}
\dot{L}&\leq-l_1\lambda_{min}(Q)\left(\|e\|-\frac{a_1}{2l_1\lambda_{min}(Q)}\right)^2+\underbrace{\frac{a_1^2}{4l_1\lambda_{min}(Q)}+a_2+\frac{\nu_1^2}{4\lambda_{min}(P)}}_\text{$\triangleq \aleph_1$}-\lambda_{min}(P)\left(\|\tilde{W}_1\|-\frac{\nu_1}{2\lambda_{min}(P)}\right)^2\\
&+\|\mathcal{S}\|\left(-\lambda_{min}(M_1)A_1\|\mathcal{S}\|^2+\underbrace{\left(b_NA_1-\lambda_{min}(M_1)b_N+\frac{1}{2}b^2_{\vartheta z}g^2_M\|R^{-1}\|+\frac{b^2_{\vartheta z}L^2_{f1}}{4l_2\lambda_{min}(Q)}\right)}_\text{$\triangleq b_1$}\|\mathcal{S}\|\right)\\
&+\underbrace{\left(b_N^2+\frac{1}{2}b^2_{\vartheta z}W_Mg_M^2\|R^{-1}\|+b_{\vartheta z}(L_{f2}+g_Mu_m)\right)}_\text{$\triangleq b_2$}\|\mathcal{S}\|
\end{split}
\end{equation}
Now, in order for $\dot{L}$ to be negative definite,
\begin{equation}
\begin{split}
-l_1&\lambda_{min}(Q)\left(\|e\|-\frac{a_1}{2l_1\lambda_{min}(Q)}\right)^2+\aleph_1<0\\
\Rightarrow&\|e\|>\frac{a_1}{2l_1\lambda_{min}(Q)}+\sqrt{\frac{\aleph_1}{l_1\lambda_{min}(Q)}}
\end{split}
\label{L2c2}
\end{equation}
or
\begin{equation}
\begin{split}
-\lambda_{min}&(P)\left(\|\tilde{W}_1\|-\frac{\nu_1}{2\lambda_{min}(P)}\right)^2+\aleph_1<0\\
\Rightarrow&\|\tilde{W}_1\|>\frac{\nu_1}{2\lambda_{min}(P)}+\sqrt{\frac{\aleph_1}{\lambda_{min}(P)}}
\end{split}
\label{eq:w1finaluub1}
\end{equation}
and
\begin{equation}
\begin{split}
-\lambda_{min}&(M_1)A_1\|\mathcal{S}\|^2+b_1\|\mathcal{S}\|+b_2<0\\
\Rightarrow&\|\mathcal{S}\|>\frac{b_1}{2A_1\lambda_{min}(M_1)}+\sqrt{\frac{b_1^2}{4A_1^2\lambda^2_{min}(M_1)}+\frac{b_2}{A_1\lambda_{min}(M_1)}}=A^{'}\\
\Rightarrow&\|\mathcal{S}\|>\frac{b_N}{\lambda_{min}(M_1)}\underbrace{\left[\frac{1}{2}\left(1-\gamma_1+\alpha_2\right)+\sqrt{\left(\frac{1}{2}(1-\gamma_1+\alpha_2)\right)^2+\gamma_1(1+\alpha_1)}\right]}_\text{$\triangleq \Gamma^{'}$}=A^{'}
\end{split}
\label{eq:gamma_dash}
\end{equation}
where, $\gamma_1=\frac{\lambda_{min}(M_1)}{A_1}$ is as defined after (\ref{eq:S}), and 
\begin{equation}
\begin{split}
 \alpha_1\triangleq\frac{\frac{1}{2}b^2_{\vartheta z}W_Mg_M^2\|R^{-1}\|+b_{\vartheta z}(L_{f2}+g_Mu_m)}{A_1b_N};~\alpha_2\triangleq\frac{\frac{1}{2}b^2_{\vartheta z}g^2_M\|R^{-1}\|+\frac{b^2_{\vartheta z}L^2_{f1}}{4l_2\lambda_{min}(Q)}}{b_N^2}
\end{split}
\label{eq:al12}
\end{equation}
Now, using the upper bound of $\mathcal{S}$ from (\ref{j_def}), the UUB bound over $\tilde{W}$ can be derived from (\ref{eq:gamma_dash}) as,
\begin{equation}
\begin{split}
\|\tilde{W}\| > \frac{\frac{b_N}{\lambda_{min}(M_1)}}{\sqrt{1+\|\varphi\|^2}}\Gamma^{'}
\end{split}
\label{wcineq}
\end{equation}

In order for $\dot{L}$ to be negative definite, 
for \textbf{Case (i)}, (\ref{ld1}) along with either (\ref{declyap}) or (\ref{eq:w1uub}) needs to hold true, similarly for \textbf{Case (ii)}, (\ref{wcineq}) alongwith either (\ref{L2c2}) or (\ref{eq:w1finaluub1}) needs to hold true.
Hence in \textbf{Case (i)}, $\tilde{W}$, $\tilde{W}_1$ and $e$ are UUB stable with corresponding bounds given in the RHS of (\ref{ld1}), (\ref{eq:w1uub}) and (\ref{declyap}), respectively.
Similarly, in \textbf{Case (ii)}, (\ref{L2c2}), (\ref{wcineq}) and (\ref{eq:w1finaluub1}) define the UUB sets for $e,\tilde{W}$ and $\tilde{W}_1$, respectively.
This completes the stability proof of the identifier and critic update mechanisms given in (\ref{update}) and (\ref{tuning_law1}), respectively.
\end{proof}

\subsection{Discussion on the update law}
\begin{remark}
Note that, the critic update law (\ref{tuning_law1}) does not require the information of approximated drift dynamics ($\hat{f}(x)$).
Since only approximated control coupling dynamics ($\hat{g}(x)$) is needed in (\ref{tuning_law1}) and policy improvement i.e., (\ref{approx_u}), the presented controller is prone to lesser approximation errors coming from identifier.
\end{remark}

\begin{remark}
\textnormal{ 
The RHS of inequalities (\ref{declyap}) and  (\ref{eq:w1uub}) are the UUB bounds for $e\text{ and }\tilde{W}_1$, respectively in \textbf{Case (i)}, and UUB (\ref{L2c2}) and (\ref{eq:w1finaluub1})
 for \textbf{Case (ii)}. 
Note that $\lambda_{min}(Q)$ and $\lambda_{min}(P)$ appear in the denominator of both $\aleph$ and $\aleph_1$. 
By suitable choice of $Q$ matrix, its minimum eigenvalue can be selected to be large. Moreover, due to experience replay in Identifier update law, the minimum eigenvalue of $P$ matrix progressively increases as new observations arrive as explained in Remark \ref{re1}.
This leads to tight UUB bounds sufficiently close to origin for $e$ and $\tilde{W}_1$.}
\end{remark}

\begin{remark}
On the other hand, the UUB bound for $\tilde{W}$ for \textbf{Case (i)} is given by (\ref{ld1}). It is evident that when the Lyapunov function is decreasing along the augmented state trajectory, variable learning rate has a direct influence over UUB bound for error in critic NN weights ($\tilde{W}$). By suitable selection of $K_1,K_2$ and $T$, the scaling factor $\Gamma$ in (\ref{eq:S}) can be varied between $.478$ and $1$ and accordingly the UUB bound of $\tilde{W}$ gets scaled down compared to that $\left(=\frac{b_N}{\lambda_{min}(M_1)}\right)$ with constant learning rate (also derived in Eq. (45) of  \cite{liu2015reinforcement}).
\end{remark}

Next, note that for \textbf{Case (ii)}  in the proof of Theorem \ref{thm:crit}, following the same form of update law but with constant learning rate in (\ref{tuning_law1}), that is with $k_2=0$ and $l=1$, the UUB bound of $\mathcal{S}$ for negative definiteness of the same Lyapunov function (as in (\ref{eq:final_L})) can be obtained as,
\begin{equation}
\begin{split}
\|\mathcal{S}\| > \frac{b_N}{\lambda_{min}(M_1)}\underbrace{\left(\frac{1+\alpha_1^{'}}{1-\alpha_2^{'}}\right)}_{\triangleq\text{$\Gamma_2^{'}$}} \Rightarrow \tilde{W} > \frac{\frac{b_N}{\lambda_{min}(M_1)}}{\sqrt{1+\|\varphi\|^2}}\Gamma_2^{'}
\end{split}
\label{const_gain1}
\end{equation}
where,
\begin{equation}
\begin{split}
\alpha_1^{'}\triangleq A_1\alpha_1;\quad\alpha_2^{'}\triangleq \frac{b_N^2\alpha_2}{\lambda_{min}(M_1)}
\end{split}
\label{eq:aldash}
\end{equation}
where, $\alpha_1,\alpha_2$ are as defined in (\ref{eq:al12}) and $\gamma_1$ is as defined after (\ref{eq:S}). 
Comparing (\ref{eq:gamma_dash}) and (\ref{const_gain1}), clearly, the UUB bound for variable gain is tighter (smaller) than that for constant gain if 
$\frac{\Gamma^{'}}{\Gamma_2^{'}}<1$, which in turn implies that,
\begin{equation}
\begin{split}
\gamma_1=\frac{\lambda_{min}(M_1)}{A_1}<\underbrace{\frac{(1+A_1\alpha_1)\left[(1+A_1\alpha_1)-(1+\alpha_2)\left(1-\frac{b_N^2\alpha_2}{\lambda_{min}(M_1)}\right)\right]}{\alpha_1\left(1-\frac{b_N^2\alpha_2}{\lambda_{min}(M_1)}\right)(1-A_1)}}_\text{$\triangleq L_1$}
\end{split}
\label{g11}
\end{equation}
Since $\gamma_1$ is a positive quantity by definition, in order to have a feasible range of $\gamma_1$, the RHS of (\ref{g11}), that is $L_1$, should be positive as well. This happens when,

\begin{equation}
\begin{split}
\alpha_2b_N^2\leq\lambda_{min}(M_1)\leq(\alpha_2+1)b_N^2
\end{split}
\label{a2}
\end{equation}


This leads to Remark \ref{last_remark}.

\begin{remark}\label{last_remark}
For \textbf{Case (ii)} in the proof of Theorem \ref{thm:crit}, that is when the Lyapunov function is increasing along augmented state trajectory, the UUB bound for $\tilde{W}$ with variable gain (\ref{wcineq}) is tighter (smaller) than that with constant learning (\ref{const_gain1}) if $K_1$ and $K_2$ in $M_1$ is such selected that (\ref{a2}) is satisfied. 
\end{remark}

\addtolength{\textheight}{-3cm}   
\section{Simulation Results}\label{res}
In this section, the identifier-critic controller structure presented above will be validated on a nonlinear spring damper system from \cite{modares2014optimal} for set point tracking. 
\begin{equation}
\begin{split}
\dot{x}_1=x_2;~~\dot{x}_2=-\frac{k}{m}x_1^3-\frac{c}{m}x_2+\frac{1}{m}u
\end{split}
\label{sys_dyn}
\end{equation}
The control saturation is considered to be $|u| \leq u_m = 2$. Physical parameters used for simulation are $m=1kg$, $k=3 N/m$, $c=.5 N-s/m$. 
The control coupling dynamics is $g=(g_1,g_2)=(0,1/m)^T$ (Note that $g_1$ here represent the first component of control coupling dynamics and not the variable gain from critic update law).
It is desired to keep the position to $x_{1d}=1m$ even in face of varying physical parameter.
Now desired velocity command is given by, $x_{2d}=\dot{x}_{1d}-5(x_1-x_{1d})$.
The control system is required to optimally drive $x_2$ to $x_{2d}$. The augmented state is, $z=(z_1,z_2)^T=(x_2-x_{2d},x_{2d})^T$ and the regressor vector for critic NN is, $\vartheta=(z_1,z_2,z_1^2,z_2^2,z_1z_2,z_1^3,z_2^3)^T$. 
The baseline learning rate was, $\alpha=20$ and exponential term in $g_1$ was $k_2=1$.

Two cases are studied, 
\begin{enumerate}
    \item Validation of without ER-based identifier and with constant learning rate gradient descent based update law for critic NN for optimal tracking of (\ref{sys_dyn}). 
    These results are depicted in Fig. \ref{fig:e_states}
    \item Validation of ER-based identifier and variable gain gradient descent-based update law for critic NN for optimal tracking of (\ref{sys_dyn}). 
    These results are shown in Fig. \ref{var_gain}.
\end{enumerate}

In order to assess the performance of the presented scheme against varying physical parameter, mass ($m$) and spring constant $k$ were varied i.e., $(m,k)=(1kg,3N/m)$ for $t<14$, then $(m,k)=(4.5kg,5N/m)$ for $14 \leq t < 30$ and then finally, $(m,k)=(8kg,9N/m)$ thereafter.
It is to be emphasised that the critic update law in IRL framework requires the information of control coupling dynamics only and not the drift dynamics. 
The approximated control coupling dynamics are $\hat{g}=(\hat{g}_1,\hat{g}_2)^T$ and augmented control coupling dynamics (refer to Eq. (\ref{eq:aug_system})) for this problem is, $\hat{G}=(\hat{g}_2,0)^T$, which is used in the generation of control effort.

From Figs. \ref{fig:identi_wout} and \ref{fig:identi} it is evident that the ER-based identifier scheme is able to make the identifier NN weights converge faster and bring the approximated NN weights closer to their ideal values than the scheme when ER is not used. 
This also leads to $\hat{g}$ settling faster and closer to its true value even in the face of parametric variations when ER-technique was utilized. This can be seen from Figs. \ref{fig:g_til_wout} and \ref{fig:g_til}, in which the difference between actual and approximated control coupling dynamics have been portrayed for respective cases. Note that the two distinct spikes that can be observed at $t=14$ and $t=30$ sec in these figures correspond to the change of mass and spring constant at those two time-instants.
 
Critic NN weights also converge faster to a tighter residual set when variable gain gradient descent is leveraged (refer to Fig \ref{fig:crit})  than when constant learning rate-based update law (i.e., $k_2=0$) is used (refer to Fig. \ref{fig:crit_wout}).
The state converge to a tighter set around the set point $(x_{1d}=1m)$ when ER-based identifier and variable gain gradient descent was utilized as can be observed from Figs. \ref{fig:states_wout} and \ref{fig:states_with}.
In both the cases, i.e., without ER-identifier and constant learning rate critic NN and with ER-identifier and variable gain gradient descent, the control policy is constrained within $\pm2 N$ (see Figs. \ref{fig:ctr} and \ref{fig:ctr_with}).

\begin{figure*}
\centering
\subcaptionbox{Critic NN weights\label{fig:crit_wout}}{\includegraphics[width=.47\textwidth,height=9.5cm,keepaspectratio,trim={.2cm 0.0cm 3.5cm .08cm},clip]{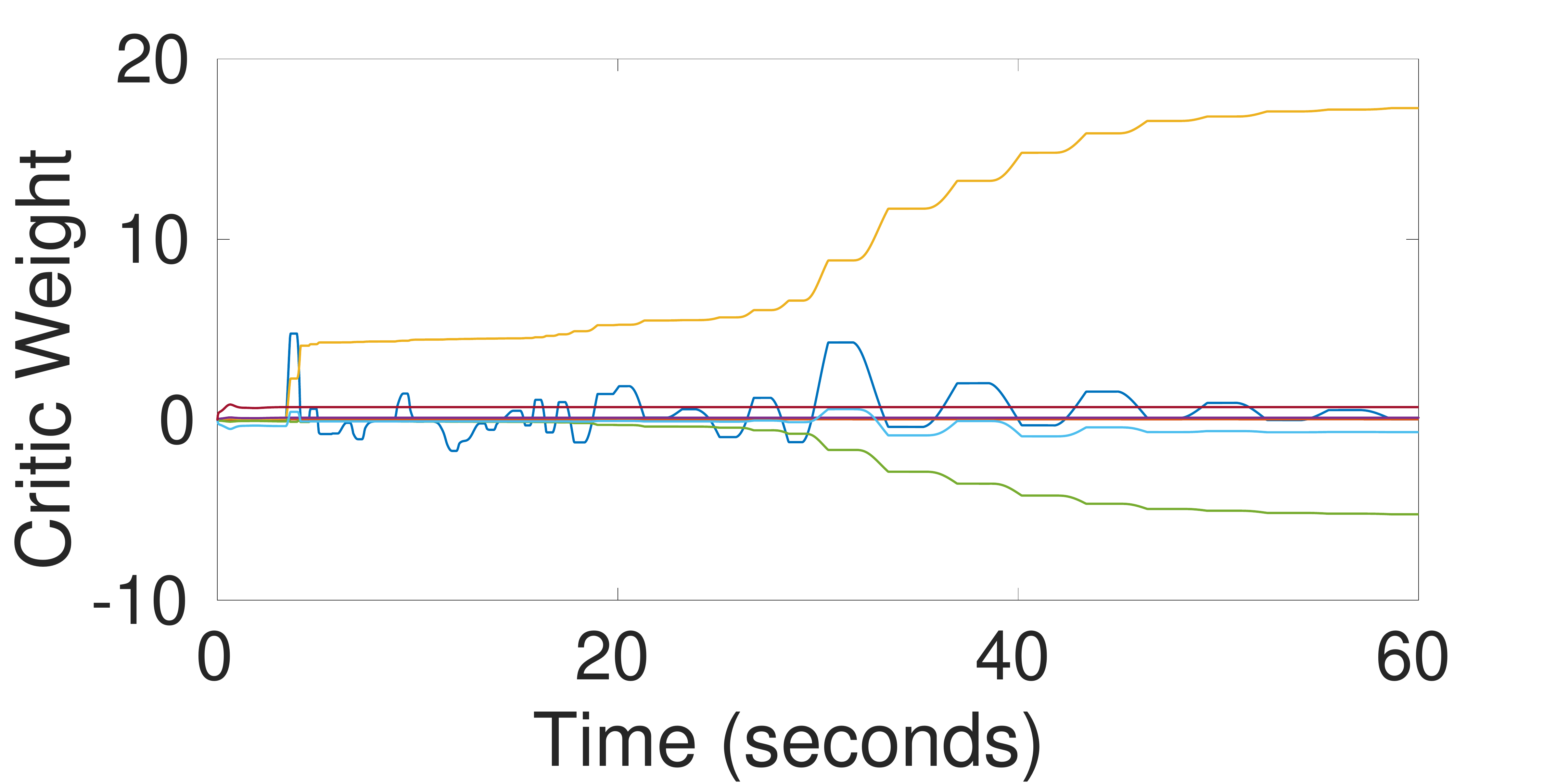}}%
\hspace{0cm} 
\subcaptionbox{Identifier NN weights\label{fig:identi_wout}}{\includegraphics[width=.47\textwidth,height=9.5cm,keepaspectratio,trim={1.8cm 0.0cm 2cm .08cm},clip]{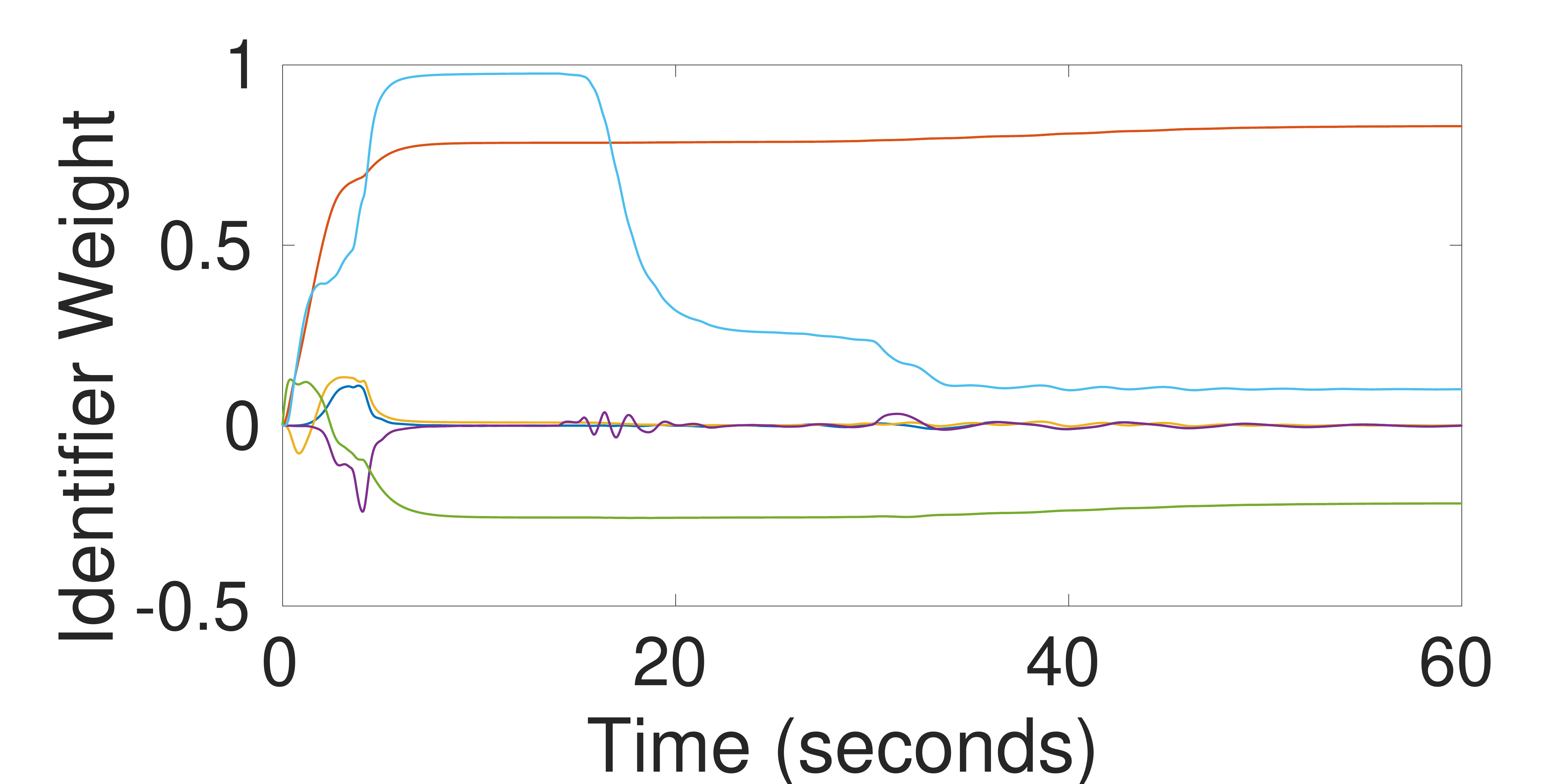}}
\hspace{0cm} 
\subcaptionbox{Difference between actual and identified control coupling dynamics \label{fig:g_til_wout}}{\includegraphics[width=.47\textwidth,height=9.5cm,keepaspectratio,trim={.01cm 0.0cm 4cm .08cm},clip]{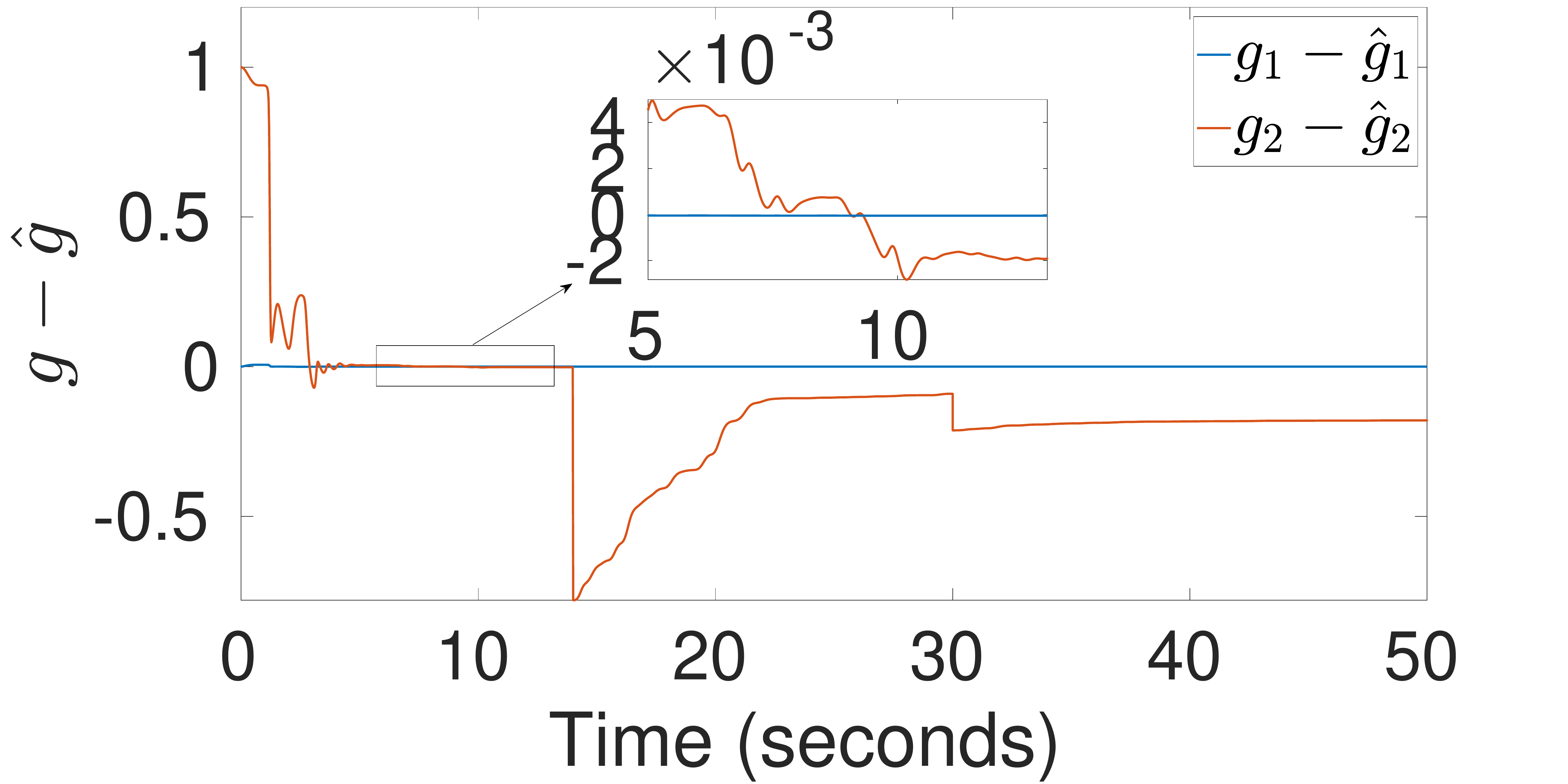}}%
\hspace{0cm} 
\subcaptionbox{States\label{fig:states_wout}}{\includegraphics[width=.47\textwidth,height=9.5cm,keepaspectratio,trim={.2cm 0.0cm 2cm .08cm},clip]{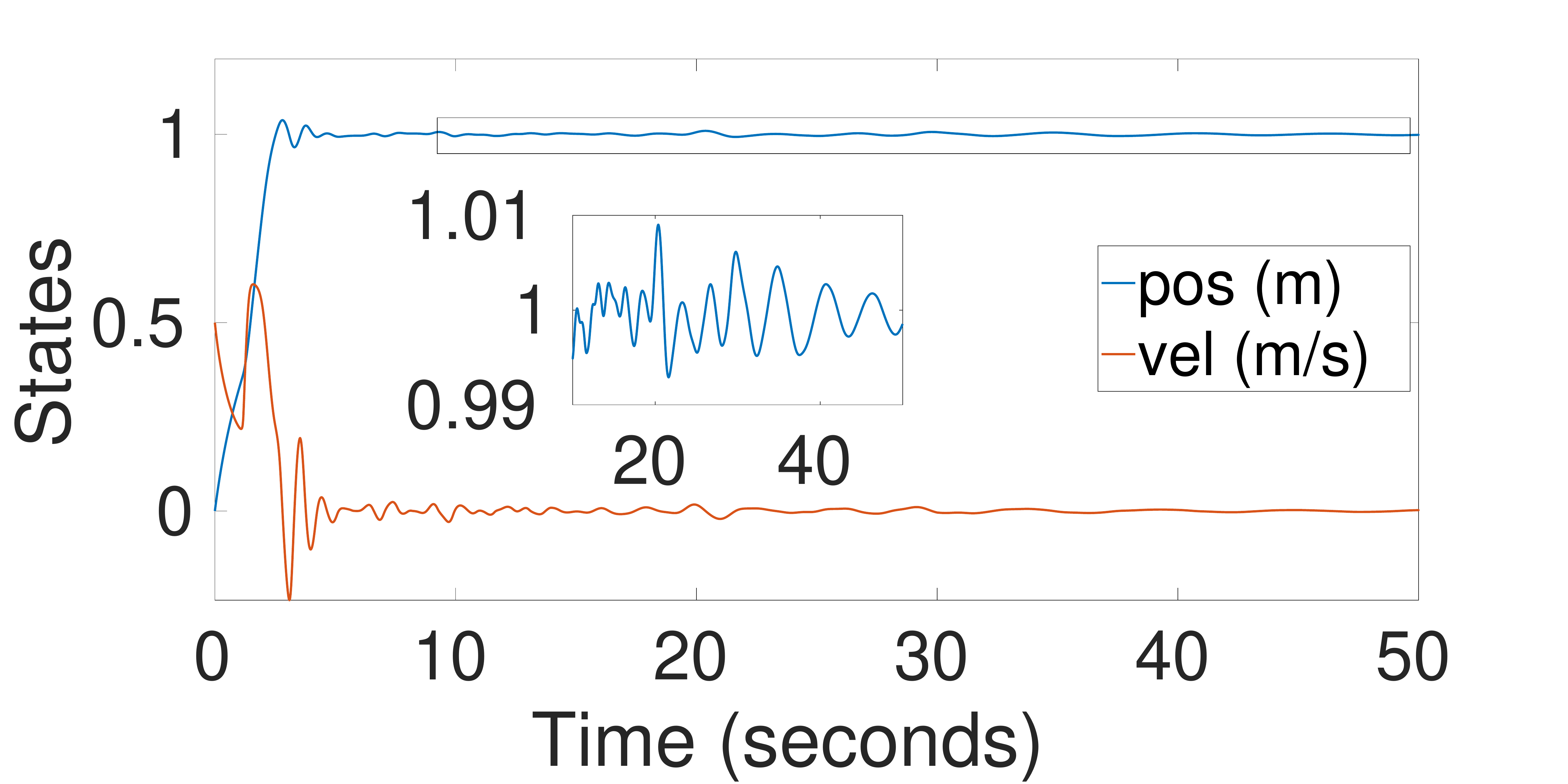}}%
\hspace{0cm} 
\subcaptionbox{Control Profile\label{fig:ctr}}{\includegraphics[width=.47\textwidth,height=9.5cm,keepaspectratio,trim={.8cm 0.0cm 2cm .08cm},clip]{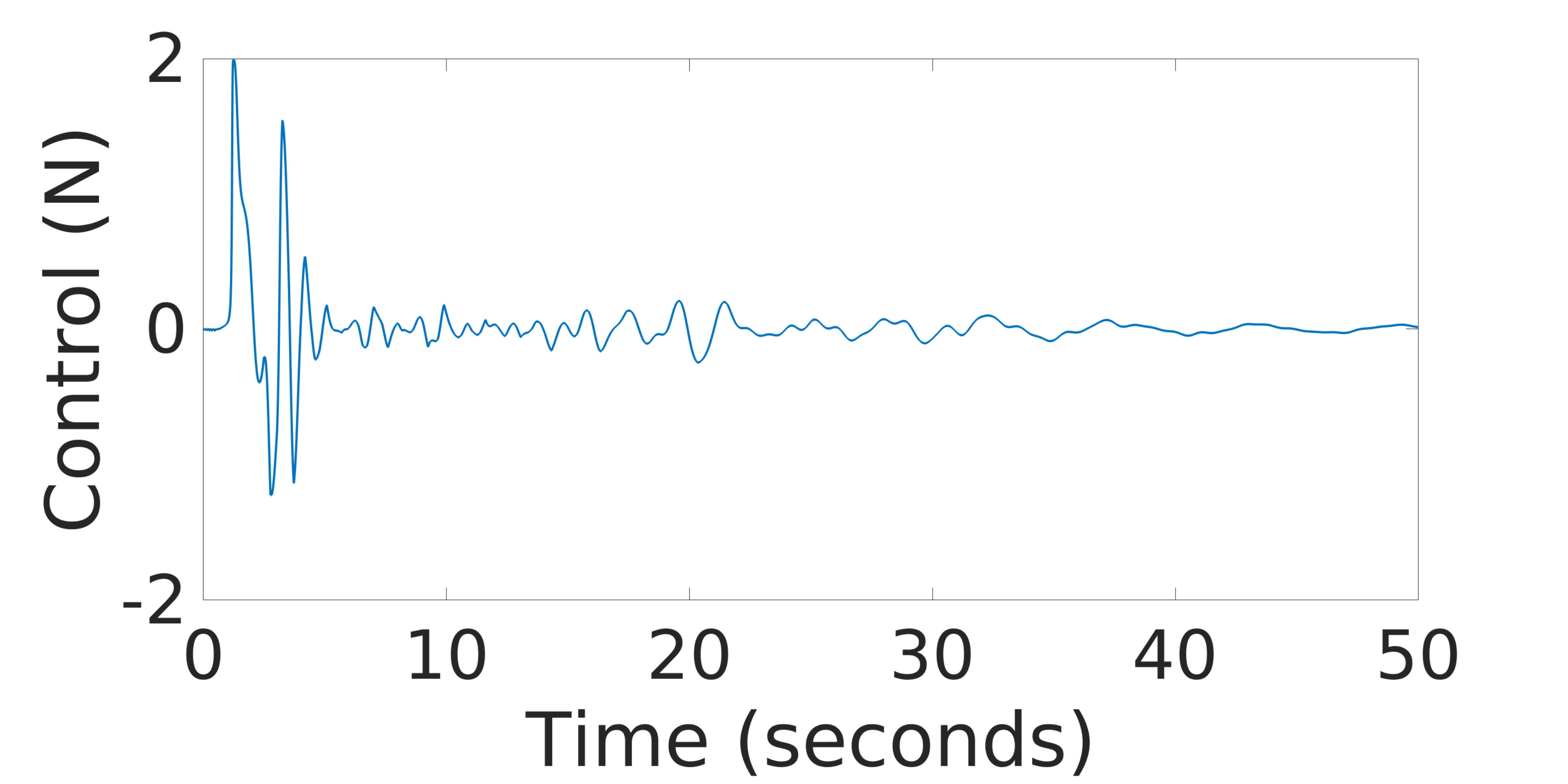}}%
\vspace{-.1cm}
\caption{Performance of the identifier-critic scheme when ER-based identifier and variable gain GD-based critic is not used}
\label{fig:e_states}
\end{figure*}
\vspace{-.2cm}

\begin{figure*}
\centering
\subcaptionbox{Critic NN weights\label{fig:crit}}{\includegraphics[width=.47\textwidth,height=9.5cm,keepaspectratio,trim={.2cm 0.0cm 2cm .08cm},clip]{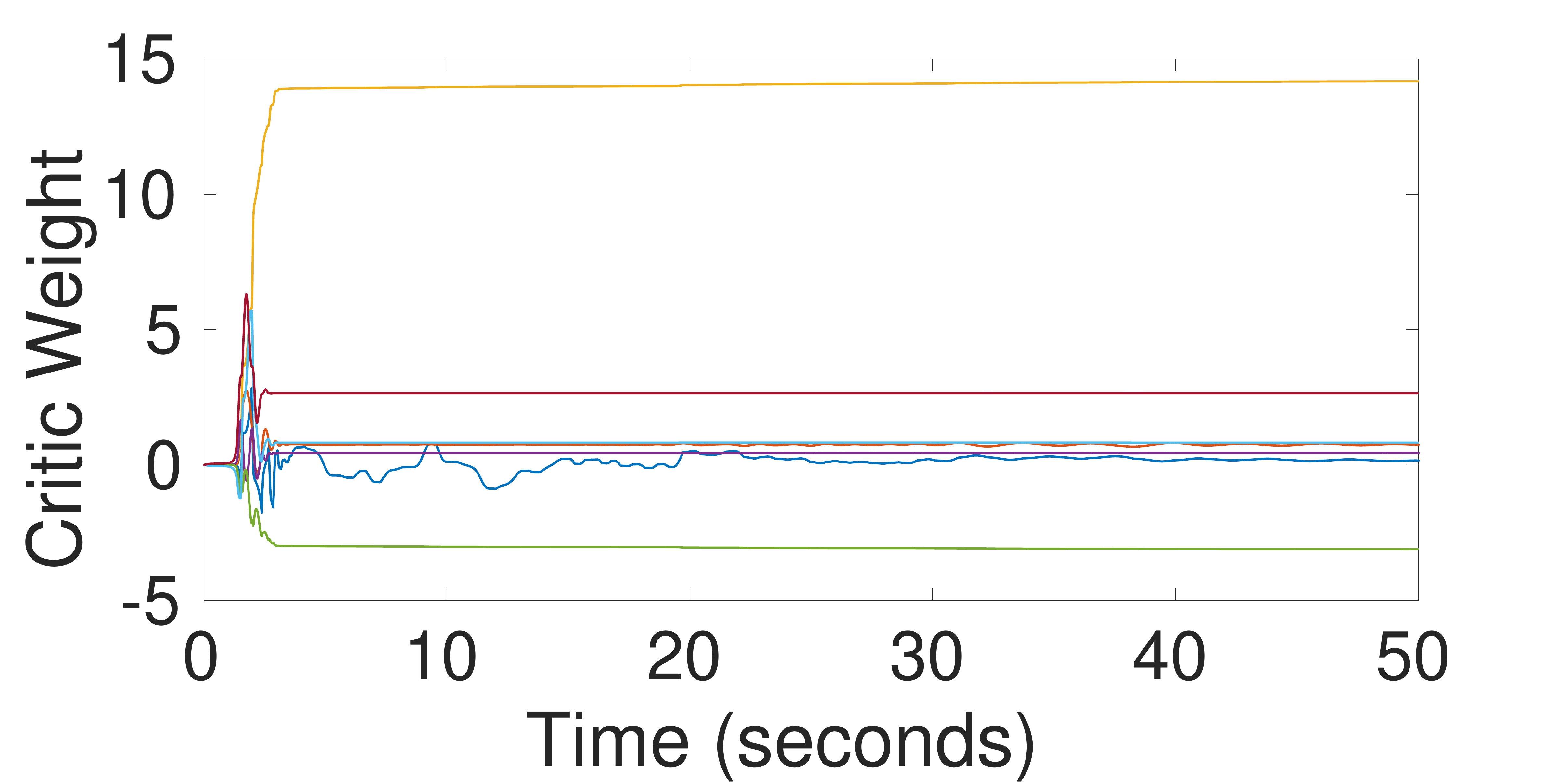}}%
\hspace{0cm}
\subcaptionbox{Identifier NN weights\label{fig:identi}}{\includegraphics[width=.47\textwidth,height=9.5cm,keepaspectratio,trim={1.8cm 0.0cm 2cm .08cm},clip]{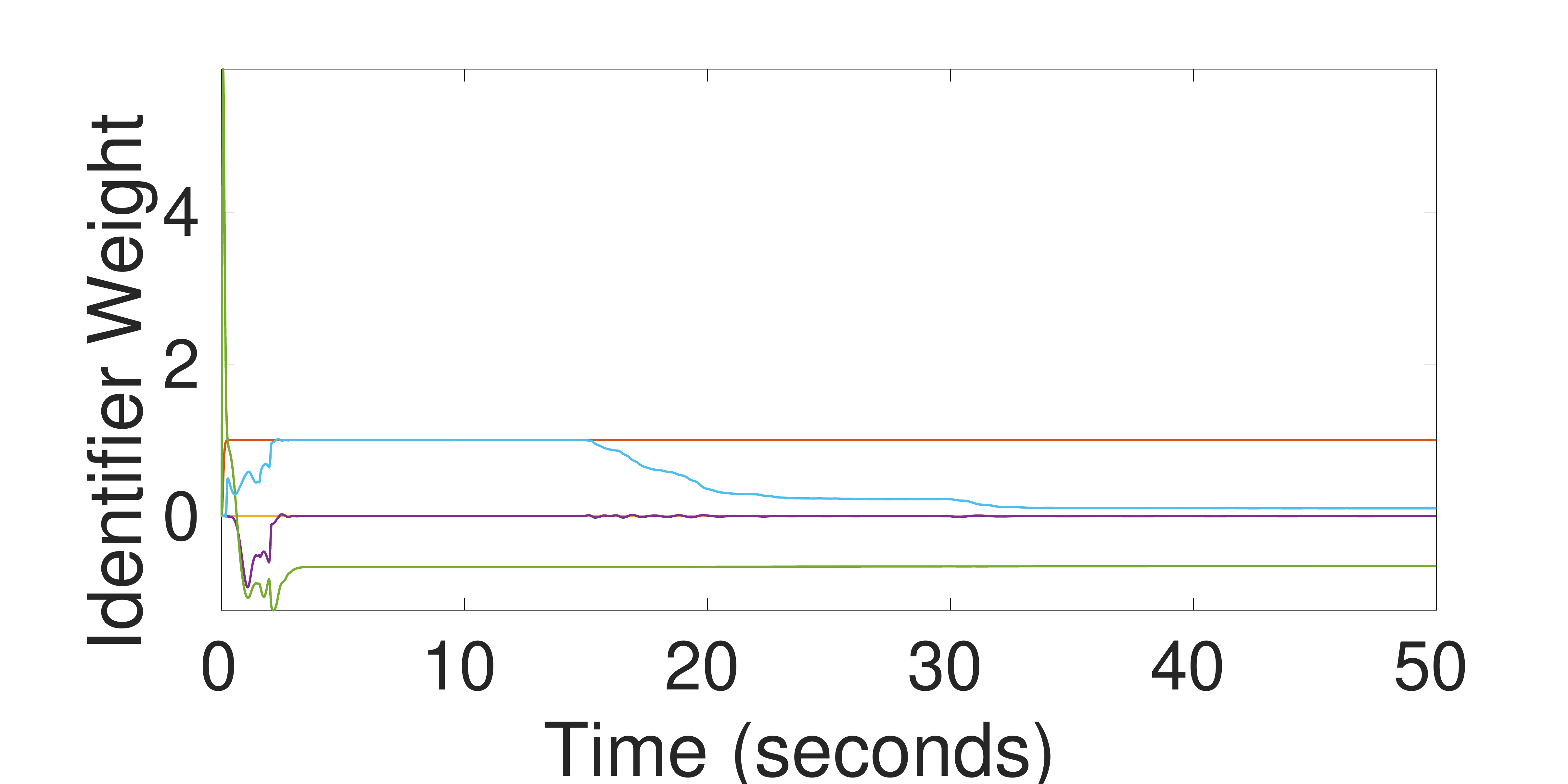}}
\hspace{0cm} 
\subcaptionbox{Difference between actual and identified control coupling dynamics\label{fig:g_til}}{\includegraphics[width=.47\textwidth,height=9.5cm,keepaspectratio,trim={0cm 0.0cm 2cm .08cm},clip]{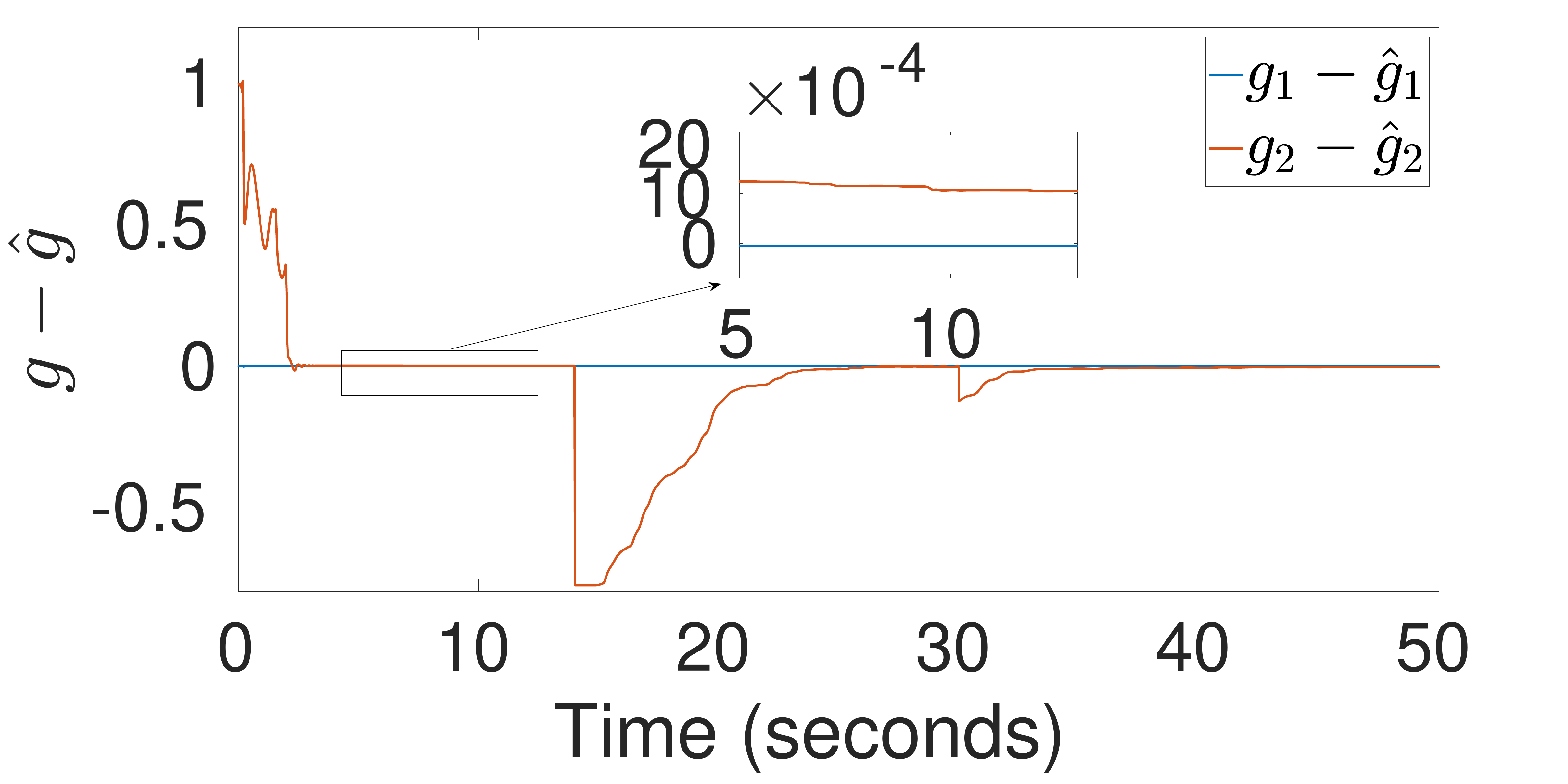}}%
\hspace{0cm} 
\subcaptionbox{States\label{fig:states_with}}{\includegraphics[width=.47\textwidth,height=9.5cm,keepaspectratio,trim={.2cm 0.0cm 2cm .08cm},clip]{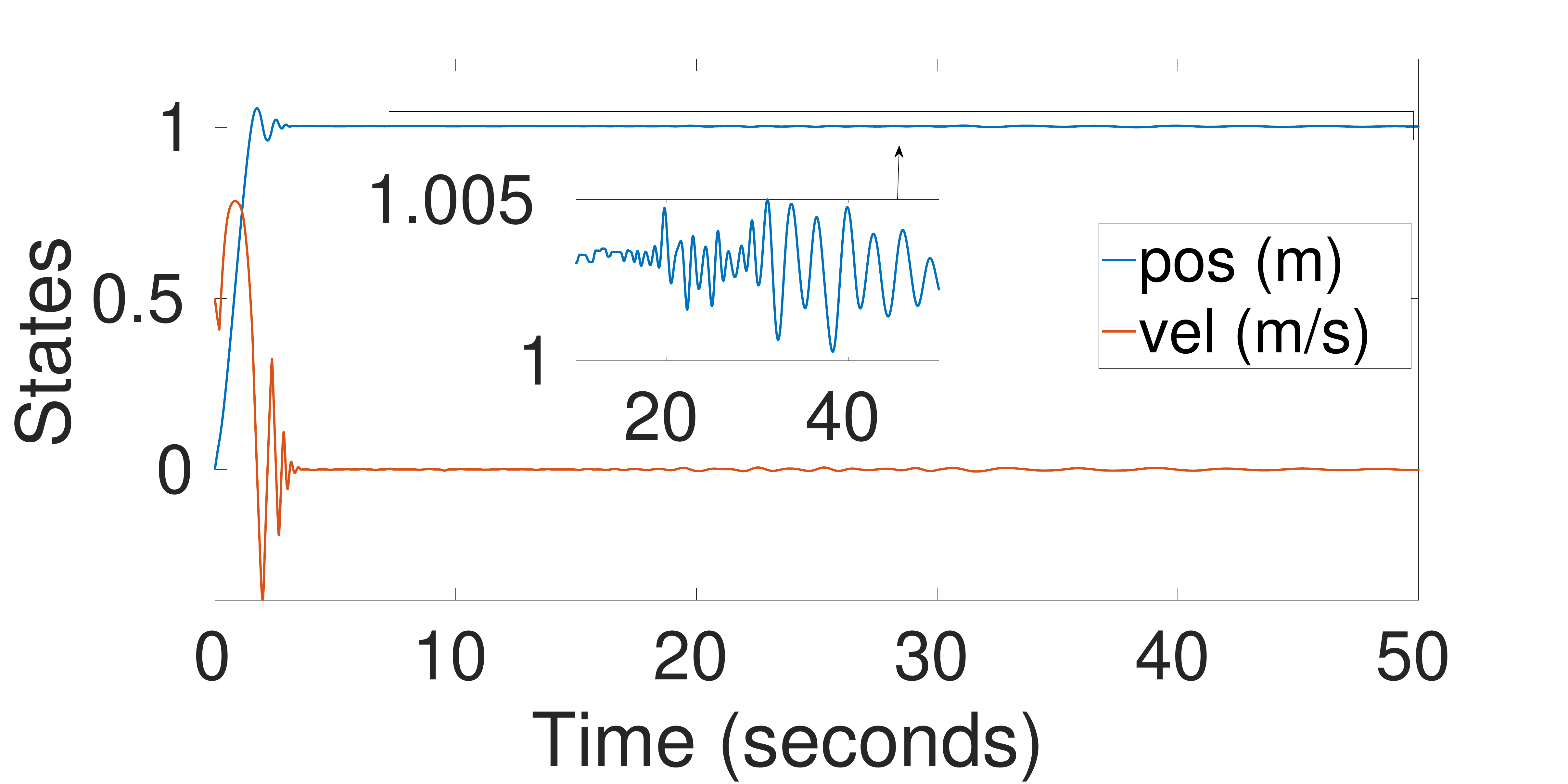}}%
\hspace{0cm} 
\subcaptionbox{Control Profile\label{fig:ctr_with}}{\includegraphics[width=.47\textwidth,height=9.5cm,keepaspectratio,trim={.8cm 0.0cm 2cm .08cm},clip]{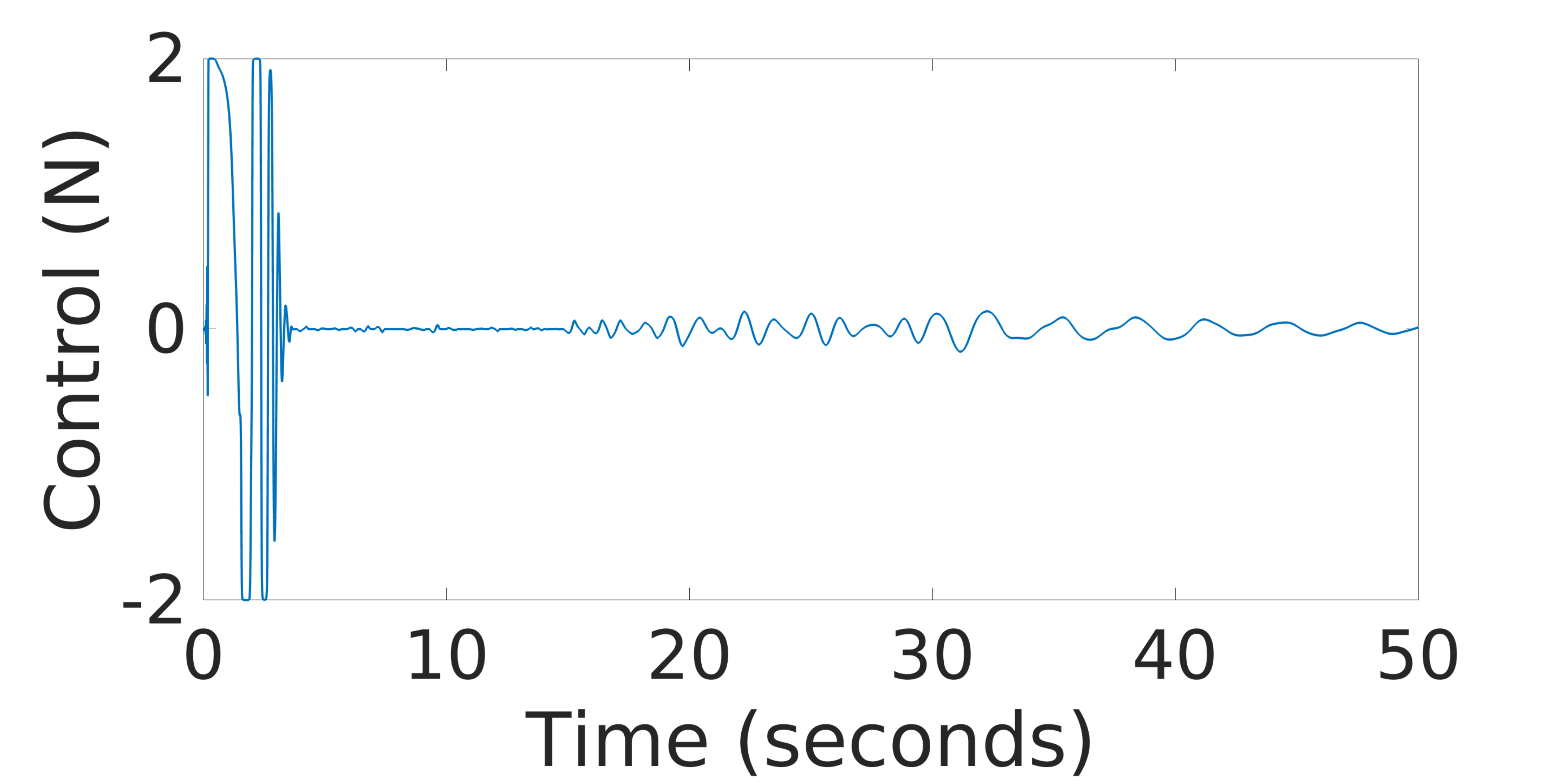}}%
\vspace{-.1cm}
\caption{Performance of the identifier-critic scheme when ER-based identifier and variable gain GD-based critic is used}
\label{var_gain}
\end{figure*}

\section{CONCLUSIONS}\label{conclusion}
Simultaneous identification and integral reinforcement learning (IRL)-based optimal tracking control of a completely unknown continuous time nonlinear system with actuator constraints is presented in this paper.
An improved system identification via experience replay (ER) technique is presented, which also helps in reducing the size of the residual set for state error and error in identifier neural network (NN) weights. The variable gain gradient descent in the presented update law could adjust the learning rate depending on the instantaneous Hamilton-Jacobi-Bellman (HJB) error. 
This results in accelerated learning when the HJB error is large and dampened learning speed when the HJI error becomes smaller.
It has an added benefit of shrinking the size of the residual set for error in critic NN weights compared to the case with constant learning rate. The presence of stabilizing term in the critic NN update law also helps in obviating the requirement of an initial stabilizing controller. The presented update law for identifier and critc NN is shown to ensure the uniform ultimate boundedness (UUB) stability of the state error and error in NN weights. The proposed framework is validated on a continuous time nonlinear system with varying physical parameters. This could also be used in general for any control-affine system application without prior knowledge of system dynamics.






\bibliographystyle{apacite}
\bibliography{sample}

\begin{thebibliography}{}

\bibitem [\protect \citeauthoryear {%
Abu-Khalaf%
\ \BBA {} Lewis%
}{%
Abu-Khalaf%
\ \BBA {} Lewis%
}{%
{\protect \APACyear {2005}}%
}]{%
abu2005nearly}
\APACinsertmetastar {%
abu2005nearly}%
\begin{APACrefauthors}%
Abu-Khalaf, M.%
\BCBT {}\ \BBA {} Lewis, F\BPBI L.%
\end{APACrefauthors}%
\unskip\
\newblock
\APACrefYearMonthDay{2005}{}{}.
\newblock
{\BBOQ}\APACrefatitle {Nearly optimal control laws for nonlinear systems with
  saturating actuators using a neural network HJB approach} {Nearly optimal
  control laws for nonlinear systems with saturating actuators using a neural
  network hjb approach}.{\BBCQ}
\newblock
\APACjournalVolNumPages{Automatica}{41}{5}{779--791}.
\PrintBackRefs{\CurrentBib}

\bibitem [\protect \citeauthoryear {%
Abu-Khalaf%
, Lewis%
\BCBL {}\ \BBA {} Huang%
}{%
Abu-Khalaf%
\ \protect \BOthers {.}}{%
{\protect \APACyear {2008}}%
}]{%
abu2008neurodynamic}
\APACinsertmetastar {%
abu2008neurodynamic}%
\begin{APACrefauthors}%
Abu-Khalaf, M.%
, Lewis, F\BPBI L.%
\BCBL {}\ \BBA {} Huang, J.%
\end{APACrefauthors}%
\unskip\
\newblock
\APACrefYearMonthDay{2008}{}{}.
\newblock
{\BBOQ}\APACrefatitle {Neurodynamic programming and zero-sum games for
  constrained control systems} {Neurodynamic programming and zero-sum games for
  constrained control systems}.{\BBCQ}
\newblock
\APACjournalVolNumPages{IEEE Transactions on Neural
  Networks}{19}{7}{1243--1252}.
\PrintBackRefs{\CurrentBib}

\bibitem [\protect \citeauthoryear {%
Bhasin%
\ \protect \BOthers {.}}{%
Bhasin%
\ \protect \BOthers {.}}{%
{\protect \APACyear {2013}}%
}]{%
bhasin2013novel}
\APACinsertmetastar {%
bhasin2013novel}%
\begin{APACrefauthors}%
Bhasin, S.%
, Kamalapurkar, R.%
, Johnson, M.%
, Vamvoudakis, K\BPBI G.%
, Lewis, F\BPBI L.%
\BCBL {}\ \BBA {} Dixon, W\BPBI E.%
\end{APACrefauthors}%
\unskip\
\newblock
\APACrefYearMonthDay{2013}{}{}.
\newblock
{\BBOQ}\APACrefatitle {A novel actor--critic--identifier architecture for
  approximate optimal control of uncertain nonlinear systems} {A novel
  actor--critic--identifier architecture for approximate optimal control of
  uncertain nonlinear systems}.{\BBCQ}
\newblock
\APACjournalVolNumPages{Automatica}{49}{1}{82--92}.
\PrintBackRefs{\CurrentBib}

\bibitem [\protect \citeauthoryear {%
Dierks%
\ \BBA {} Jagannathan%
}{%
Dierks%
\ \BBA {} Jagannathan%
}{%
{\protect \APACyear {2010}}%
}]{%
dierks2010optimal}
\APACinsertmetastar {%
dierks2010optimal}%
\begin{APACrefauthors}%
Dierks, T.%
\BCBT {}\ \BBA {} Jagannathan, S.%
\end{APACrefauthors}%
\unskip\
\newblock
\APACrefYearMonthDay{2010}{}{}.
\newblock
{\BBOQ}\APACrefatitle {Optimal control of affine nonlinear continuous-time
  systems} {Optimal control of affine nonlinear continuous-time
  systems}.{\BBCQ}
\newblock
\BIn{} \APACrefbtitle {Proceedings of the 2010 American Control Conference}
  {Proceedings of the 2010 american control conference}\ (\BPGS\ 1568--1573).
\PrintBackRefs{\CurrentBib}

\bibitem [\protect \citeauthoryear {%
Finlayson%
}{%
Finlayson%
}{%
{\protect \APACyear {2013}}%
}]{%
finlayson2013method}
\APACinsertmetastar {%
finlayson2013method}%
\begin{APACrefauthors}%
Finlayson, B\BPBI A.%
\end{APACrefauthors}%
\unskip\
\newblock
\APACrefYear{2013}.
\newblock
\APACrefbtitle {The method of weighted residuals and variational principles}
  {The method of weighted residuals and variational principles}\ (\BVOL~73).
\newblock
\APACaddressPublisher{}{SIAM}.
\PrintBackRefs{\CurrentBib}

\bibitem [\protect \citeauthoryear {%
Hou%
, Na%
, Lv%
, Gao%
\BCBL {}\ \BBA {} Wu%
}{%
Hou%
\ \protect \BOthers {.}}{%
{\protect \APACyear {2017}}%
}]{%
hou2017adaptive}
\APACinsertmetastar {%
hou2017adaptive}%
\begin{APACrefauthors}%
Hou, D.%
, Na, J.%
, Lv, Y.%
, Gao, G.%
\BCBL {}\ \BBA {} Wu, X.%
\end{APACrefauthors}%
\unskip\
\newblock
\APACrefYearMonthDay{2017}{}{}.
\newblock
{\BBOQ}\APACrefatitle {Adaptive optimal tracking control for continuous-time
  systems using identifier-critic based dynamic programming} {Adaptive optimal
  tracking control for continuous-time systems using identifier-critic based
  dynamic programming}.{\BBCQ}
\newblock
\BIn{} \APACrefbtitle {2017 36th Chinese Control Conference (CCC)} {2017 36th
  chinese control conference (ccc)}\ (\BPGS\ 2583--2588).
\PrintBackRefs{\CurrentBib}

\bibitem [\protect \citeauthoryear {%
Kiumarsi%
, Lewis%
, Modares%
, Karimpour%
\BCBL {}\ \BBA {} Naghibi-Sistani%
}{%
Kiumarsi%
\ \protect \BOthers {.}}{%
{\protect \APACyear {2014}}%
}]{%
kiumarsi2014reinforcement}
\APACinsertmetastar {%
kiumarsi2014reinforcement}%
\begin{APACrefauthors}%
Kiumarsi, B.%
, Lewis, F\BPBI L.%
, Modares, H.%
, Karimpour, A.%
\BCBL {}\ \BBA {} Naghibi-Sistani, M\BHBI B.%
\end{APACrefauthors}%
\unskip\
\newblock
\APACrefYearMonthDay{2014}{}{}.
\newblock
{\BBOQ}\APACrefatitle {Reinforcement Q-learning for optimal tracking control of
  linear discrete-time systems with unknown dynamics} {Reinforcement q-learning
  for optimal tracking control of linear discrete-time systems with unknown
  dynamics}.{\BBCQ}
\newblock
\APACjournalVolNumPages{Automatica}{50}{4}{1167--1175}.
\PrintBackRefs{\CurrentBib}

\bibitem [\protect \citeauthoryear {%
C.~Liu%
, Zhang%
, Ren%
\BCBL {}\ \BBA {} Liang%
}{%
C.~Liu%
\ \protect \BOthers {.}}{%
{\protect \APACyear {2019}}%
}]{%
liu2019analysis}
\APACinsertmetastar {%
liu2019analysis}%
\begin{APACrefauthors}%
Liu, C.%
, Zhang, H.%
, Ren, H.%
\BCBL {}\ \BBA {} Liang, Y.%
\end{APACrefauthors}%
\unskip\
\newblock
\APACrefYearMonthDay{2019}{}{}.
\newblock
{\BBOQ}\APACrefatitle {An Analysis of IRL-Based Optimal Tracking Control of
  Unknown Nonlinear Systems with Constrained Input} {An analysis of irl-based
  optimal tracking control of unknown nonlinear systems with constrained
  input}.{\BBCQ}
\newblock
\APACjournalVolNumPages{Neural Processing Letters}{}{}{1--20}.
\PrintBackRefs{\CurrentBib}

\bibitem [\protect \citeauthoryear {%
D.~Liu%
, Yang%
, Wang%
\BCBL {}\ \BBA {} Wei%
}{%
D.~Liu%
\ \protect \BOthers {.}}{%
{\protect \APACyear {2015}}%
}]{%
liu2015reinforcement}
\APACinsertmetastar {%
liu2015reinforcement}%
\begin{APACrefauthors}%
Liu, D.%
, Yang, X.%
, Wang, D.%
\BCBL {}\ \BBA {} Wei, Q.%
\end{APACrefauthors}%
\unskip\
\newblock
\APACrefYearMonthDay{2015}{}{}.
\newblock
{\BBOQ}\APACrefatitle {Reinforcement-learning-based robust controller design
  for continuous-time uncertain nonlinear systems subject to input constraints}
  {Reinforcement-learning-based robust controller design for continuous-time
  uncertain nonlinear systems subject to input constraints}.{\BBCQ}
\newblock
\APACjournalVolNumPages{IEEE transactions on cybernetics}{45}{7}{1372--1385}.
\PrintBackRefs{\CurrentBib}

\bibitem [\protect \citeauthoryear {%
Lv%
, Na%
, Yang%
, Wu%
\BCBL {}\ \BBA {} Guo%
}{%
Lv%
\ \protect \BOthers {.}}{%
{\protect \APACyear {2016}}%
}]{%
lv2016online}
\APACinsertmetastar {%
lv2016online}%
\begin{APACrefauthors}%
Lv, Y.%
, Na, J.%
, Yang, Q.%
, Wu, X.%
\BCBL {}\ \BBA {} Guo, Y.%
\end{APACrefauthors}%
\unskip\
\newblock
\APACrefYearMonthDay{2016}{}{}.
\newblock
{\BBOQ}\APACrefatitle {Online adaptive optimal control for continuous-time
  nonlinear systems with completely unknown dynamics} {Online adaptive optimal
  control for continuous-time nonlinear systems with completely unknown
  dynamics}.{\BBCQ}
\newblock
\APACjournalVolNumPages{International Journal of Control}{89}{1}{99--112}.
\PrintBackRefs{\CurrentBib}

\bibitem [\protect \citeauthoryear {%
Lyashevskiy%
}{%
Lyashevskiy%
}{%
{\protect \APACyear {1996}}%
}]{%
lyashevskiy1996constrained}
\APACinsertmetastar {%
lyashevskiy1996constrained}%
\begin{APACrefauthors}%
Lyashevskiy, S.%
\end{APACrefauthors}%
\unskip\
\newblock
\APACrefYearMonthDay{1996}{}{}.
\newblock
{\BBOQ}\APACrefatitle {Constrained optimization and control of nonlinear
  systems: new results in optimal control} {Constrained optimization and
  control of nonlinear systems: new results in optimal control}.{\BBCQ}
\newblock
\BIn{} \APACrefbtitle {Proceedings of 35th IEEE Conference on Decision and
  Control} {Proceedings of 35th ieee conference on decision and control}\
  (\BVOL~1, \BPGS\ 541--546).
\PrintBackRefs{\CurrentBib}

\bibitem [\protect \citeauthoryear {%
Mishra%
\ \BBA {} Ghosh%
}{%
Mishra%
\ \BBA {} Ghosh%
}{%
{\protect \APACyear {2019}}%
{\protect \APACexlab {{\protect \BCnt {1}}}}}]{%
mishra2019criticonly}
\APACinsertmetastar {%
mishra2019criticonly}%
\begin{APACrefauthors}%
Mishra, A.%
\BCBT {}\ \BBA {} Ghosh, S.%
\end{APACrefauthors}%
\unskip\
\newblock
\APACrefYearMonthDay{2019{\protect \BCnt {1}}}{}{}.
\newblock
{\BBOQ}\APACrefatitle {Critic-Only Integral Reinforcement Learning Driven by
  Variable Gain Gradient Descent for Optimal Tracking Control} {Critic-only
  integral reinforcement learning driven by variable gain gradient descent for
  optimal tracking control}.{\BBCQ}
\newblock
\APACjournalVolNumPages{arXiv preprint arXiv:1911.04153}{}{}{}.
\PrintBackRefs{\CurrentBib}

\bibitem [\protect \citeauthoryear {%
Mishra%
\ \BBA {} Ghosh%
}{%
Mishra%
\ \BBA {} Ghosh%
}{%
{\protect \APACyear {2019}}%
{\protect \APACexlab {{\protect \BCnt {2}}}}}]{%
mishra2019variable}
\APACinsertmetastar {%
mishra2019variable}%
\begin{APACrefauthors}%
Mishra, A.%
\BCBT {}\ \BBA {} Ghosh, S.%
\end{APACrefauthors}%
\unskip\
\newblock
\APACrefYearMonthDay{2019{\protect \BCnt {2}}}{}{}.
\newblock
{\BBOQ}\APACrefatitle {Variable Gain Gradient Descent-based Robust
  Reinforcement Learning for Optimal Tracking Control of Unknown Nonlinear
  System with Input-Constraints} {Variable gain gradient descent-based robust
  reinforcement learning for optimal tracking control of unknown nonlinear
  system with input-constraints}.{\BBCQ}
\newblock
\APACjournalVolNumPages{arXiv preprint arXiv:1911.04157}{}{}{}.
\PrintBackRefs{\CurrentBib}

\bibitem [\protect \citeauthoryear {%
Mishra%
\ \BBA {} Ghosh%
}{%
Mishra%
\ \BBA {} Ghosh%
}{%
{\protect \APACyear {2020}}%
}]{%
mishra2020texthinfty}
\APACinsertmetastar {%
mishra2020texthinfty}%
\begin{APACrefauthors}%
Mishra, A.%
\BCBT {}\ \BBA {} Ghosh, S.%
\end{APACrefauthors}%
\unskip\
\newblock
\APACrefYearMonthDay{2020}{}{}.
\newblock
{\BBOQ}\APACrefatitle {$\text{H}_{\infty}$ Tracking Control via Variable Gain
  Gradient Descent-Based Integral Reinforcement Learning for Unknown Continuous
  Time Nonlinear System} {$\text{H}_{\infty}$ tracking control via variable
  gain gradient descent-based integral reinforcement learning for unknown
  continuous time nonlinear system}.{\BBCQ}
\newblock
\APACjournalVolNumPages{arXiv preprint arXiv:2001.07355}{}{}{}.
\PrintBackRefs{\CurrentBib}

\bibitem [\protect \citeauthoryear {%
Modares%
\ \BBA {} Lewis%
}{%
Modares%
\ \BBA {} Lewis%
}{%
{\protect \APACyear {2014}}%
}]{%
modares2014optimal}
\APACinsertmetastar {%
modares2014optimal}%
\begin{APACrefauthors}%
Modares, H.%
\BCBT {}\ \BBA {} Lewis, F\BPBI L.%
\end{APACrefauthors}%
\unskip\
\newblock
\APACrefYearMonthDay{2014}{}{}.
\newblock
{\BBOQ}\APACrefatitle {Optimal tracking control of nonlinear partially-unknown
  constrained-input systems using integral reinforcement learning} {Optimal
  tracking control of nonlinear partially-unknown constrained-input systems
  using integral reinforcement learning}.{\BBCQ}
\newblock
\APACjournalVolNumPages{Automatica}{50}{7}{1780--1792}.
\PrintBackRefs{\CurrentBib}

\bibitem [\protect \citeauthoryear {%
Modares%
, Lewis%
\BCBL {}\ \BBA {} Jiang%
}{%
Modares%
\ \protect \BOthers {.}}{%
{\protect \APACyear {2015}}%
}]{%
modares2015h}
\APACinsertmetastar {%
modares2015h}%
\begin{APACrefauthors}%
Modares, H.%
, Lewis, F\BPBI L.%
\BCBL {}\ \BBA {} Jiang, Z\BHBI P.%
\end{APACrefauthors}%
\unskip\
\newblock
\APACrefYearMonthDay{2015}{}{}.
\newblock
{\BBOQ}\APACrefatitle {${H}_{\infty}$ Tracking Control of Completely Unknown
  Continuous-Time Systems via Off-Policy Reinforcement Learning}
  {${H}_{\infty}$ tracking control of completely unknown continuous-time
  systems via off-policy reinforcement learning}.{\BBCQ}
\newblock
\APACjournalVolNumPages{IEEE transactions on neural networks and learning
  systems}{26}{10}{2550--2562}.
\PrintBackRefs{\CurrentBib}

\bibitem [\protect \citeauthoryear {%
Modares%
, Lewis%
\BCBL {}\ \BBA {} Naghibi-Sistani%
}{%
Modares%
\ \protect \BOthers {.}}{%
{\protect \APACyear {2013}}%
}]{%
modares2013adaptive}
\APACinsertmetastar {%
modares2013adaptive}%
\begin{APACrefauthors}%
Modares, H.%
, Lewis, F\BPBI L.%
\BCBL {}\ \BBA {} Naghibi-Sistani, M\BHBI B.%
\end{APACrefauthors}%
\unskip\
\newblock
\APACrefYearMonthDay{2013}{}{}.
\newblock
{\BBOQ}\APACrefatitle {Adaptive optimal control of unknown constrained-input
  systems using policy iteration and neural networks} {Adaptive optimal control
  of unknown constrained-input systems using policy iteration and neural
  networks}.{\BBCQ}
\newblock
\APACjournalVolNumPages{IEEE Transactions on Neural Networks and Learning
  Systems}{24}{10}{1513--1525}.
\PrintBackRefs{\CurrentBib}

\bibitem [\protect \citeauthoryear {%
Modares%
, Lewis%
\BCBL {}\ \BBA {} Naghibi-Sistani%
}{%
Modares%
\ \protect \BOthers {.}}{%
{\protect \APACyear {2014}}%
}]{%
modares2014integral}
\APACinsertmetastar {%
modares2014integral}%
\begin{APACrefauthors}%
Modares, H.%
, Lewis, F\BPBI L.%
\BCBL {}\ \BBA {} Naghibi-Sistani, M\BHBI B.%
\end{APACrefauthors}%
\unskip\
\newblock
\APACrefYearMonthDay{2014}{}{}.
\newblock
{\BBOQ}\APACrefatitle {Integral reinforcement learning and experience replay
  for adaptive optimal control of partially-unknown constrained-input
  continuous-time systems} {Integral reinforcement learning and experience
  replay for adaptive optimal control of partially-unknown constrained-input
  continuous-time systems}.{\BBCQ}
\newblock
\APACjournalVolNumPages{Automatica}{50}{1}{193--202}.
\PrintBackRefs{\CurrentBib}

\bibitem [\protect \citeauthoryear {%
Na%
, Lv%
, Wu%
, Guo%
\BCBL {}\ \BBA {} Chen%
}{%
Na%
\ \protect \BOthers {.}}{%
{\protect \APACyear {2014}}%
}]{%
na2014approximate}
\APACinsertmetastar {%
na2014approximate}%
\begin{APACrefauthors}%
Na, J.%
, Lv, Y.%
, Wu, X.%
, Guo, Y.%
\BCBL {}\ \BBA {} Chen, Q.%
\end{APACrefauthors}%
\unskip\
\newblock
\APACrefYearMonthDay{2014}{}{}.
\newblock
{\BBOQ}\APACrefatitle {Approximate optimal tracking control for continuous-time
  unknown nonlinear systems} {Approximate optimal tracking control for
  continuous-time unknown nonlinear systems}.{\BBCQ}
\newblock
\BIn{} \APACrefbtitle {Proceedings of the 33rd Chinese Control Conference}
  {Proceedings of the 33rd chinese control conference}\ (\BPGS\ 8990--8995).
\PrintBackRefs{\CurrentBib}

\bibitem [\protect \citeauthoryear {%
Ren%
, Lewis%
\BCBL {}\ \BBA {} Zhang%
}{%
Ren%
\ \protect \BOthers {.}}{%
{\protect \APACyear {2009}}%
}]{%
ren2009neural}
\APACinsertmetastar {%
ren2009neural}%
\begin{APACrefauthors}%
Ren, X.%
, Lewis, F\BPBI L.%
\BCBL {}\ \BBA {} Zhang, J.%
\end{APACrefauthors}%
\unskip\
\newblock
\APACrefYearMonthDay{2009}{}{}.
\newblock
{\BBOQ}\APACrefatitle {Neural network compensation control for mechanical
  systems with disturbances} {Neural network compensation control for
  mechanical systems with disturbances}.{\BBCQ}
\newblock
\APACjournalVolNumPages{Automatica}{45}{5}{1221--1226}.
\PrintBackRefs{\CurrentBib}

\bibitem [\protect \citeauthoryear {%
Vamvoudakis%
\ \BBA {} Lewis%
}{%
Vamvoudakis%
\ \BBA {} Lewis%
}{%
{\protect \APACyear {2010}}%
}]{%
vamvoudakis2010online}
\APACinsertmetastar {%
vamvoudakis2010online}%
\begin{APACrefauthors}%
Vamvoudakis, K\BPBI G.%
\BCBT {}\ \BBA {} Lewis, F\BPBI L.%
\end{APACrefauthors}%
\unskip\
\newblock
\APACrefYearMonthDay{2010}{}{}.
\newblock
{\BBOQ}\APACrefatitle {Online actor--critic algorithm to solve the
  continuous-time infinite horizon optimal control problem} {Online
  actor--critic algorithm to solve the continuous-time infinite horizon optimal
  control problem}.{\BBCQ}
\newblock
\APACjournalVolNumPages{Automatica}{46}{5}{878--888}.
\PrintBackRefs{\CurrentBib}

\bibitem [\protect \citeauthoryear {%
Vamvoudakis%
, Vrabie%
\BCBL {}\ \BBA {} Lewis%
}{%
Vamvoudakis%
\ \protect \BOthers {.}}{%
{\protect \APACyear {2014}}%
}]{%
vamvoudakis2014online}
\APACinsertmetastar {%
vamvoudakis2014online}%
\begin{APACrefauthors}%
Vamvoudakis, K\BPBI G.%
, Vrabie, D.%
\BCBL {}\ \BBA {} Lewis, F\BPBI L.%
\end{APACrefauthors}%
\unskip\
\newblock
\APACrefYearMonthDay{2014}{}{}.
\newblock
{\BBOQ}\APACrefatitle {Online adaptive algorithm for optimal control with
  integral reinforcement learning} {Online adaptive algorithm for optimal
  control with integral reinforcement learning}.{\BBCQ}
\newblock
\APACjournalVolNumPages{International Journal of Robust and Nonlinear
  Control}{24}{17}{2686--2710}.
\PrintBackRefs{\CurrentBib}

\bibitem [\protect \citeauthoryear {%
Vrabie%
\ \BBA {} Lewis%
}{%
Vrabie%
\ \BBA {} Lewis%
}{%
{\protect \APACyear {2009}}%
}]{%
vrabie2009neural}
\APACinsertmetastar {%
vrabie2009neural}%
\begin{APACrefauthors}%
Vrabie, D.%
\BCBT {}\ \BBA {} Lewis, F.%
\end{APACrefauthors}%
\unskip\
\newblock
\APACrefYearMonthDay{2009}{}{}.
\newblock
{\BBOQ}\APACrefatitle {Neural network approach to continuous-time direct
  adaptive optimal control for partially unknown nonlinear systems} {Neural
  network approach to continuous-time direct adaptive optimal control for
  partially unknown nonlinear systems}.{\BBCQ}
\newblock
\APACjournalVolNumPages{Neural Networks}{22}{3}{237--246}.
\PrintBackRefs{\CurrentBib}

\bibitem [\protect \citeauthoryear {%
Yang%
, Liu%
\BCBL {}\ \BBA {} Wang%
}{%
Yang%
\ \protect \BOthers {.}}{%
{\protect \APACyear {2014}}%
}]{%
yang2014reinforcement}
\APACinsertmetastar {%
yang2014reinforcement}%
\begin{APACrefauthors}%
Yang, X.%
, Liu, D.%
\BCBL {}\ \BBA {} Wang, D.%
\end{APACrefauthors}%
\unskip\
\newblock
\APACrefYearMonthDay{2014}{}{}.
\newblock
{\BBOQ}\APACrefatitle {Reinforcement learning for adaptive optimal control of
  unknown continuous-time nonlinear systems with input constraints}
  {Reinforcement learning for adaptive optimal control of unknown
  continuous-time nonlinear systems with input constraints}.{\BBCQ}
\newblock
\APACjournalVolNumPages{International Journal of Control}{87}{3}{553--566}.
\PrintBackRefs{\CurrentBib}

\bibitem [\protect \citeauthoryear {%
Yang%
, Liu%
\BCBL {}\ \BBA {} Wei%
}{%
Yang%
\ \protect \BOthers {.}}{%
{\protect \APACyear {2015}}%
}]{%
yang2015robust}
\APACinsertmetastar {%
yang2015robust}%
\begin{APACrefauthors}%
Yang, X.%
, Liu, D.%
\BCBL {}\ \BBA {} Wei, Q.%
\end{APACrefauthors}%
\unskip\
\newblock
\APACrefYearMonthDay{2015}{}{}.
\newblock
{\BBOQ}\APACrefatitle {Robust tracking control of uncertain nonlinear systems
  using adaptive dynamic programming} {Robust tracking control of uncertain
  nonlinear systems using adaptive dynamic programming}.{\BBCQ}
\newblock
\BIn{} \APACrefbtitle {International Conference on Neural Information
  Processing} {International conference on neural information processing}\
  (\BPGS\ 9--16).
\PrintBackRefs{\CurrentBib}

\bibitem [\protect \citeauthoryear {%
Zhang%
, Cui%
, Zhang%
\BCBL {}\ \BBA {} Luo%
}{%
Zhang%
\ \protect \BOthers {.}}{%
{\protect \APACyear {2011}}%
}]{%
zhang2011data}
\APACinsertmetastar {%
zhang2011data}%
\begin{APACrefauthors}%
Zhang, H.%
, Cui, L.%
, Zhang, X.%
\BCBL {}\ \BBA {} Luo, Y.%
\end{APACrefauthors}%
\unskip\
\newblock
\APACrefYearMonthDay{2011}{}{}.
\newblock
{\BBOQ}\APACrefatitle {Data-driven robust approximate optimal tracking control
  for unknown general nonlinear systems using adaptive dynamic programming
  method} {Data-driven robust approximate optimal tracking control for unknown
  general nonlinear systems using adaptive dynamic programming method}.{\BBCQ}
\newblock
\APACjournalVolNumPages{IEEE Transactions on Neural
  Networks}{22}{12}{2226--2236}.
\PrintBackRefs{\CurrentBib}

\bibitem [\protect \citeauthoryear {%
Zhang%
, Cui%
, Luo%
\BCBL {}\ \BBA {} Jiang%
}{%
Zhang%
\ \protect \BOthers {.}}{%
{\protect \APACyear {2017}}%
}]{%
zhang2017finite}
\APACinsertmetastar {%
zhang2017finite}%
\begin{APACrefauthors}%
Zhang, H.%
, Cui, X.%
, Luo, Y.%
\BCBL {}\ \BBA {} Jiang, H.%
\end{APACrefauthors}%
\unskip\
\newblock
\APACrefYearMonthDay{2017}{}{}.
\newblock
{\BBOQ}\APACrefatitle {Finite-Horizon ${H}_{\infty}$ Tracking Control for
  Unknown Nonlinear Systems With Saturating Actuators} {Finite-horizon
  ${H}_{\infty}$ tracking control for unknown nonlinear systems with saturating
  actuators}.{\BBCQ}
\newblock
\APACjournalVolNumPages{IEEE transactions on neural networks and learning
  systems}{29}{4}{1200--1212}.
\PrintBackRefs{\CurrentBib}

\bibitem [\protect \citeauthoryear {%
Zhu%
, Zhao%
\BCBL {}\ \BBA {} Li%
}{%
Zhu%
\ \protect \BOthers {.}}{%
{\protect \APACyear {2016}}%
}]{%
zhu2016using}
\APACinsertmetastar {%
zhu2016using}%
\begin{APACrefauthors}%
Zhu, Y.%
, Zhao, D.%
\BCBL {}\ \BBA {} Li, X.%
\end{APACrefauthors}%
\unskip\
\newblock
\APACrefYearMonthDay{2016}{}{}.
\newblock
{\BBOQ}\APACrefatitle {Using reinforcement learning techniques to solve
  continuous-time non-linear optimal tracking problem without system dynamics}
  {Using reinforcement learning techniques to solve continuous-time non-linear
  optimal tracking problem without system dynamics}.{\BBCQ}
\newblock
\APACjournalVolNumPages{IET Control Theory \&
  Applications}{10}{12}{1339--1347}.
\PrintBackRefs{\CurrentBib}

\end{thebibliography}

\appendix\label{Appendix}
\section{Lemma}\label{L1}
\begin{lemma}\label{cu}
If $\psi^{-1}$ is monotonic odd and increasing, and $R_i>0$, where $u_i \in \mathbb{R},i=1,2,...,m$, then following inequality holds true:
\begin{equation}
 C(u_i)=2u_m\int_{0}^{u_i}\psi^{-1}(\frac{\nu}{u_m})R_id\nu   \geq 0
\end{equation}
\end{lemma}
\begin{proof}
If $\psi^{-1}$ is monotonic odd and increasing, then,
\begin{equation}
\begin{split}
 \Big(\frac{\nu}{u_m}\Big)\psi^{-1}\Big(\frac{\nu}{u_m}\Big) \geq 0
\end{split}
\end{equation}
or 
\begin{equation}
\nu\psi^{-1}\Big(\frac{\nu}{u_m}\Big) \geq 0
\end{equation}
where $\nu\in \mathbb{R}$ and $u_m>0$. Let $\theta=1/u_m$. 
In order to prove that, $2u_m\int_{0}^{u_i}\psi^{-1}(\nu/u_m)R_id\nu \geq 0$, it is enough to prove that, $\int_{0}^{u_i}\psi^{-1}(\nu\theta)d\nu   \geq 0$. 
In order to prove this inequality, a variable, $\mathcal{K}\in[0,\theta]$ is assumed.
Therefore,
\begin{equation}
\begin{split}
\int_{0}^{u_i}\psi^{-1}(\nu\theta)d\nu=\frac{1}{\theta}\int_{0}^{u_i\theta}\psi^{-1}(l)dl
\end{split}
\end{equation}
where $l=\nu\theta$. Similarly, 
\begin{equation}
\begin{split}
 \frac{1}{\theta}\int_{0}^{u_i\theta}\psi^{-1}(l)dl=\frac{1}{\theta}\int_{0}^{\theta}\psi^{-1}(u_i\mathcal{K})u_id\mathcal{K}
\end{split}
\end{equation}
by utilizing $l=u_i\mathcal{K}$\\
Since, $\psi^{-1}(u_i\mathcal{K})u_i \geq 0$, which implies, 
\begin{equation}
\frac{1}{\theta}\int_{0}^{\theta}\psi^{-1}(u_i\mathcal{K})u_id\mathcal{K}\geq 0
\end{equation}
\end{proof}

\begin{lemma}\label{cu1}
Following equality holds true,
\begin{equation}
\begin{split}
 2u_m\int_0^{-u_m\tanh{A(z)}}\tanh^{-1}(\nu/u_m)^TRd\nu=2u_m^2A^T(z)R\tanh{A(z)}+
u_m^2\sum_{i=1}^{m}R_i\ln[1-\tanh^2{A_i(z)}]   
\end{split}
\end{equation}
\begin{proof}
\begin{equation}
\begin{split}
\int \tanh^{-1}\Big(\frac{x}{a}\Big)=\frac{1}{2}a\ln{(a^2-x^2)}+x\tanh^{-1}\Big(\frac{x}{a}\Big)+I
\end{split}
\end{equation}
Therefore, 
\begin{equation}
\begin{split}
&\int_0^u \tanh^{-1} \Big(\frac{\nu}{u_m}\Big)d\nu=\frac{1}{2}u_m\ln{(u_m^2-\nu^2)}+\nu\tanh^{-1}\Big(\frac{\nu}{u_m}\Big)\Big|_0^u\\
&2u_m\int_0^u \tanh^{-1} \Big(\frac{\nu}{u_m}\Big)d\nu=u_m^2\ln{(u_m^2-\nu^2)}+2u_m\nu\tanh^{-1}\Big(\frac{\nu}{u_m}\Big)\Big|_0^u\\
&=u_m^2\ln{(1-\frac{u^2}{u_m^2})}+2u_m^2\tanh{A(z)}\\
&=u_m^2\ln{(1-\tanh^2{A(z)})}+2u_m^2\tanh{A(z)}\\
\end{split}
\end{equation}
where, $u=-u_m\tanh{A(z)}$ is scalar. Now if $u$ is a vector, then,
\begin{equation}
\small
\begin{split}
 2u_m\int_0^u \tanh^{-1} \Big(\frac{\nu}{u_m}\Big)Rd\nu=2u_m^2A^T(z)R\tanh{A(z)}+
u_m^2\sum_{i=1}^{m}R_i\ln[1-\tanh^2{A_i(z)}]
\end{split}
\end{equation}
\end{proof}
\end{lemma}

\begin{lemma}\label{mean_val_lem}
Following equation holds true,
\begin{equation}
\begin{split}
u=-u_m\tanh{\Big(\frac{1}{2u_m}R^{-1}\hat{G}^T \nabla{\vartheta}^TW+\epsilon_{uu}\Big)}=-u_m\tanh{(\tau_1(z))}+\epsilon_{u}
\end{split}
\label{ustar}
\end{equation}
where, $\small \epsilon_{uu}=(1/2u_m)R^{-1}\hat{G}^T\nabla{\varepsilon}(z)=\small [\varepsilon_{{uu}_{11}},\varepsilon_{{uu}_{12}},...,\varepsilon_{{uu}_{1m}}]^T \in \mathbb{R}^m$.
$\tau_{1}(z)=(1/2u_m)R^{-1}\hat{G}^T \nabla{\vartheta}^TW=[\tau_{11},...,\tau_{1m}]^T \in \mathbb{R}^m$ and 
$\epsilon_{u}=-(1/2)((I_m-diag(\tanh^2{(q)}))R^{-1}\hat{G}^T\nabla{\epsilon})$ 
with $q \in \mathbb{R}^m$ and $q_i \in \mathbb{R}$ considered between $\tau_{1i}+\varepsilon_{uui}$ and $\epsilon_{uui}$ i.e., $i^{th}$ element of $\epsilon_{uu}$.
\end{lemma}
\begin{proof}
\begin{equation}
u=-u_m\tanh{(\tau_1+\varepsilon_{uu})}
\end{equation}
Using mean value theorem,
\begin{equation}
\tanh{(\tau_1+\varepsilon_{uu})}-\tanh{(\tau_1)}=\tanh^{'}{(q)}\varepsilon_{uu}=(I_m-diag(\tanh^2{(q)}))\varepsilon_{uu}
\label{mmt}
\end{equation}
where, $q \in \mathbb{R}^m$ and $q_i \in \mathbb{R}$ lying between $\tau_{1i}$ and $\tau_{1i}+\varepsilon_{uui}$. 

Now, using the expression for $\varepsilon_{uu}$ in (\ref{mmt}), $\tanh{(\tau_1+\varepsilon_{uu})}$ can be rewritten as,
\begin{equation}
\tanh{(\tau_1+\varepsilon_{uu})}=\tanh{(\tau_1)}+(I_m-diag(\tanh^2{(q)}))\left(\frac{1}{2u_m}R^{-1}\hat{G}^T\nabla{\varepsilon}(z)\right)
\end{equation}
Multiplying $-u_m$ on both sides,
\begin{equation}
-u_m\tanh{(\tau_1+\varepsilon_{uu})}=-u_m\tanh{(\tau_1)}-\frac{1}{2}(I_m-diag(\tanh^2{(q)}))\left(R^{-1}\hat{G}^T(z)\nabla{\varepsilon}(z)\right)
\end{equation}
Hence proved.
\end{proof}

\begin{lemma}\label{tanhlem}
Following vector inequality holds true:\\
\begin{equation}
\|\tanh{(\tau_1(z))}-\tanh{(\tau_2(z))}\| \leq T_m \leq 2\sqrt{m}
\end{equation}
where $T_m=\sqrt{\sum_{i=1}^mmin(|\tau_{1i}-\tau_{2i}|^2,4)}$, $\tau_1(z)$ and $\tau_2(z)$ both belong in $\mathbb{R}^m$, therefore, $\tanh{(\tau_i(z))} \in \mathbb{R}^m,~i=1,2$. \\
\textbf{Proof}:
Since, $\tanh{(.)}$ is 1-Lipschitz, one can write,
\begin{equation}
\begin{split}
|\tanh{(\tau_{1i})}-\tanh{(\tau_{2i})}| \leq |\tau_{1i}-\tau_{2i}|
\end{split}
\end{equation}
Therefore using the above inequality and the fact that, $-1\leq \tanh{(.)} \leq 1$
\begin{equation}
\begin{split}
\|\tanh{(\tau_1(z))}-\tanh{(\tau_2(z))}\|^2 &= \sum_{i=1}^m|\tanh{\tau_{1i}}-\tanh{\tau_{2i}}|^2\\
& \leq \sum_{i=1}^mmin(|\tau_{1i}-\tau_{2i}|,2)^2\\
&\leq \sum_{i=1}^mmin(|\tau_{1i}-\tau_{2i}|^2,4)
\end{split}
\end{equation}
One can also see, using the absolute upper bound of $\tanh{(.)}$. 
\begin{equation}
\sum_{i=1}^mmin(|\tau_{1i}-\tau_{2i}|^2,4)\leq 2\sqrt{m}
\end{equation}
Which implies, 
\begin{equation}
\begin{split}
\|\tanh{(\tau_1(z))}-\tanh{(\tau_2(z))}\|\leq T_m \leq 2\sqrt{m} 
\end{split}
\end{equation}
\end{lemma}






\end{document}